\documentclass[journal]{IEEEtran}
\usepackage{bm}
\usepackage{verbatim}
\usepackage{amsfonts}
\usepackage{amssymb}
\usepackage{stfloats}
\usepackage{cite}
\usepackage{graphicx}
\usepackage{psfrag}
\usepackage{subfigure}
\usepackage{amsmath}
\usepackage{url}
\usepackage{array}
\usepackage{epstopdf}
\usepackage{authblk}
\usepackage{graphicx} 
\usepackage{mdframed}
\usepackage{amsthm}         
\usepackage{lipsum}
\usepackage{verbatim} 
\usepackage{authblk}
\usepackage{mathtools}
\usepackage{cuted}
\usepackage{authblk}
\usepackage{booktabs}
\usepackage{subfigure}
\usepackage{siunitx}
\usepackage{algorithm, algorithmic}
\usepackage{mathtools}
\usepackage{soul}

\newtheorem{Lemma}{Lemma}

\newtheorem{Remark}{Remark}

\usepackage{graphicx} % Required for inserting images

\begin{document}
\title{Parametric-Sensitivity Aware Retransmission for Efficient AI Downloading}

\author{{You~Zhou,~\IEEEmembership{Graduate Student Member,~IEEE}, 
{Qunsong Zeng,~\IEEEmembership{Member,~IEEE}}, 
and {Kaibin Huang,~\IEEEmembership{Fellow,~IEEE}}}
\vspace{-3mm}
\thanks{Y. Zhou, Q. Zeng and K. Huang are with the Department of Electrical and Electronic Engineering, The University of Hong Kong, Hong Kong SAR, China (Email: zhouyou@eee.hku.hk, qszeng@eee.hku.hk, huangkb@hku.hk). Corresponding authors: Q. Zeng; K. Huang.}}

\maketitle
\begin{abstract}
    The edge artificial intelligence (AI) applications in next-generation mobile networks demand efficient AI-model downloading techniques to support real-time, on-device inference. 
    However, transmitting high-dimensional AI models over wireless channels remains challenging due to limited communication resources. 
    To address this issue, we propose a parametric-sensitivity-aware retransmission (PASAR) framework that manages radio-resource usage of different parameter packets according to their importance on model inference accuracy, known as parametric sensitivity.
    Empirical analysis reveals a highly right-skewed sensitivity distribution, indicating that only a small fraction of parameters significantly affect model performance. 
    Leveraging this insight, we design a novel online retransmission protocol, i.e., the PASAR protocol, that adaptively terminates packet transmission based on real-time bit error rate (BER) measurements and the associated parametric sensitivity. 
    The protocol employs an adaptive, round-wise stopping criterion, enabling heterogeneous, packet-level retransmissions that preserve overall model functionality but reduce overall latency. 
    Extensive experiments across diverse deep neural network architectures and real-world datasets demonstrate that PASAR substantially outperforms classical hybrid automatic repeat request (HARQ) schemes in terms of communication efficiency and latency.
\end{abstract}


\section{Introduction}
The next-generation mobile network is expected to deliver artificial intelligence (AI) capabilities at the network edge, enabling user-specific services with low latency and high efficiency \cite{tong20226g}. 
Among the three use cases of integrated AI and communication as identified in 3GPP, one involves downloading AI models to end-user devices for on-device inference~\cite{3gpp}. 
Compared with server-based inference, on-device inference reduces end-to-end latency by avoiding uplink data transmission and preserves privacy by keeping raw data locally.
However, due to limited radio resources, wireless transmission of high-dimensional AI models becomes a critical bottleneck~\cite{wu2024efficient}. 
This challenge highlights the need for advanced communication protocols and resource management strategies that optimize transmission efficiency while maintaining the downloaded AI model's performance. 
Traditional communication techniques assume that all bits or symbols are equally important, leading to design objectives centered on maximizing transmission rate. 
While effective in general-purpose data transmission, this approach overlooks the unique characteristics of AI model parameters. 
In particular, model parameters exhibit varying degrees of sensitivity: some parameters, when distorted by noise, degrade the inference performance more significantly than others \cite{yan2024hessian}. 
This insight motivates a novel design approach: parametric-sensitivity-aware resource allocation. 
In this work, we leverage this principle to redesign the conventional automatic repeat request (ARQ) protocol for improving the efficiency of AI model downloading in resource-constrained wireless communication systems.

Edge AI primarily addresses the challenge of data-model separation, where user data is generated or collected on edge devices while AI models are stored on edge servers~\cite{letaief2021edge}. 
To address this issue, two main paradigms have emerged. 
The first approach, known as edge inference, involves uploading data (or extracted features) from the edge device to the edge server for processing, and the inference results are subsequently sent back to the device \cite{shao2021communication}. 
Such an approach is particularly advantageous given that the model is too large to be executed on edge devices, such as in the case of large language models (LLMs) or other complex models with high computational requirements. 
However, it heavily relies on the stability and reliability of the wireless connection, as it necessitates frequent uplink transmissions of data (or features) and timely downlink delivery of inference results~\cite{zhou2025communication}. 
This dependency can be problematic for mission-critical applications that require real-time responses, such as augmented/virtual reality (AR/VR)~\cite{ARVR} or remote healthcare~\cite{remotehealthcare}, particularly under poor wireless channel conditions. 
Moreover, transmitting raw data (or features) over the air raises significant concerns regarding data privacy and security, as sensitive information may be exposed to potential interception or unauthorized access~\cite{wu2024efficient}. 
The growing computational power of modern AI processors equipped at edge devices has spurred interest in the second approach: AI model downloading for local inference~\cite{huang2023situ}. 
In this paradigm, the edge server dynamically selects and transmits a task-specific AI model to the edge device, enabling on-device inference tailored to the user's needs~\cite{wu2024efficient}. 
This approach effectively addresses both real-time issues and privacy concerns inherent in the first approach \cite{hou2021model}. 
Furthermore, such a paradigm ensures real-time AI model updates for devices such that devices can promptly adapt to changes on user needs or context~\cite{huang2023situ}.

Despite its advantages, real-time AI model downloading still faces a communication bottleneck, primarily due to the high dimensionality of the model parameters required for accurate inference~\cite{ahmed2021hybrid}. 
This challenge is further exacerbated in the scenario envisioned by 3GPP, where edge servers must efficiently handle diverse and dynamic inference requests from multiple users~\cite{reinforcechallenges}. 
To address the issues, it is essential to develop communication-efficient strategies for AI model downloading that ensure both timely delivery and reliable on-device inference.  
One promising research direction lies in radio resource management~\cite{wen2019overview}. 
Aligning with this direction, we explore retransmission as a foundational time-allocation strategy to control the transmission reliability under noisy channel conditions. 
ARQ protocols, a retransmission scheme ubiquitously deployed in modern wireless systems, offer operational simplicity and reliability without requiring channel state information at the transmitter (CSIT), which can be difficult to obtain in practice\cite{ahmed2021hybrid}. 
Variants such as hybrid ARQ (HARQ) enhance this process by combining forward error correction with ARQ to improve error resilience~\cite{sheng2017performance, nadeem2021nonorthogonal}. 
However, traditional communication techniques are typically designed with an emphasis on either bit-level reliability or data rate maximization~\cite{ratereliabe}; they lack awareness of AI tasks and their characteristics. Therefore, the achieved efficiency is sub-optimal for large-scale model downloading. 
To fill the gap, we propose incorporating the heterogeneity of the importance levels of the model parameters into making retransmission decisions. 

The importance-aware retransmission approaches optimize radio resource allocation in edge AI systems by prioritizing the delivery of important data, features, or parameters. 
In the scenario of edge learning, data importance is investigated through importance measures in terms of uncertainty and diversity~\cite{huang2013active}. The resultant retransmission protocols preferentially acquire highly uncertain data so as to accelerate model convergence~\cite{liu2020wireless}. 
In the scenario of split inference, the metric of feature importance is introduced based on discriminant gain~\cite{bayeserror}. 
This metric helps to set packet priorities in progressive feature transmission, and thereby enhances the communication efficiency~\cite{lan2022progressive}. 
For parametric importance in the current scenario of AI downloading, the most relevant literature is sensitivity analysis, where the \emph{sensitivity} refers to the degree to which variations in individual model parameters affect the model’s output~\cite{yan2024hessian}. It can be quantified by the Hessians of the learning loss function with respect to individual model parameters. Their values can be evaluated and stored at the end of the training process~\cite{yan2022swim}.
Highly sensitive parameters are more critical than others, since perturbations to the former can lead to significant degradation in the inference accuracy. In the area of computing, parametric sensitivity has been explored in several domains, including neural network pruning, quantization, and model compression~\cite{molchanov2019importance}. 
However, the exploitation of sensitivity to boost communication efficiency is an uncharted area.

To bridge the gap, we propose in this work a parametric-sensitivity-aware retransmission (PASAR) protocol to enable efficient AI model downloading in resource-constrained systems. 
The key idea is to incorporate parameter sensitivity into retransmission decisions, ensuring that communication resources are used where they matter most.
Specifically, unlike traditional ARQ/HARQ schemes that treat all packets equally, the PASAR protocol adaptively allocates time resources to the transmission of parameter packets according to their sensitivity levels. By effectively leveraging the highly skewed sensitivity distribution, the design can significantly reduce retransmission latency without compromising the inference accuracy.

The key contributions and findings of this work are summarized as follows.
\begin{itemize}
    \item \textbf{Downloading loss analysis:} 
    We develop a Hessian-based performance metric called downloading loss to quantify the effect of bit errors in individual model parameters on the loss function. This metric extends the conventional definition\cite{yan2022swim}, to account for channel bit errors. Using this metric, we observe highly right-skewed sensitivity distributions across representative model architectures.
    These findings challenge the conventional uniform reliability paradigm in ARQ/HARQ schemes and provide theoretical justification for parametric sensitivity-based designs.
    \item \textbf{Sensitivity aware retransmission control:} 
    Based on the proceeding metric of downloading loss, we propose the PASAR protocol for efficient AI model downloading. The design approach abstracts the end-to-end retransmission process as an online multiple-choice knapsack problem (MCKP), aiming at downloading latency minimization subject to a global loss constraint.
    At the core of the protocol is an \emph{on-device retransmission control} module using stopping thresholds as in the MCKP framework. The novelty of the design lies in an algorithm for adapting thresholds based on real-time BER observations and parametric sensitivity.
    As a result, PASAR focuses radio resources on transmitting the small subset of high-sensitivity parameters to reduce latency without compromising on-device inference accuracy.
    \item \textbf{Extension to parameter pruning:}
    We discuss the extension of the PASAR protocol to investigate model pruning prior to downloading. A fundamental trade-off that underpins the design is identified. Specifically, while pruning decreases transmission latency by reducing the number of transmitted parameters, a too aggressive pruning can also compromise inference performance. Even though pruning may reduce the performance gains achievable with PASAR by reducing sensitivity skewness, PASAR continues to outperform conventional channel-aware ARQ schemes across different pruning ratios as shown by experiments.
    \item \textbf{Experiments:} 
    We validate the effectiveness of PASAR through benchmarking against conventional ARQ/HARQ schemes. The experimental results demonstrate consistent latency reduction under a constraint on downloading loss.
    However, as the sensitivity distribution approaches uniform due to aggressive pruning, this advantage diminishes and PASAR gradually converges to the conventional uniform retransmission schemes.
\end{itemize}
The remainder of this paper is organized as follows.
Section II presents the system model and performance metrics for AI model downloading. Section III analyzes the parametric sensitivity and explores its distribution in representative AI models. Section IV formulates the corresponding online retransmission control problem and presents the overview of PASAR protocol. Section V designs the PASAR's retransmission-control module. Section VI discusses an extension of PASAR with parameter pruning. The experimental results are evaluated in Section VII, and concluding remarks are given in Section VIII.

\vspace{-4mm}
\section{System and Metrics}
We consider an AI model downloading system comprising a server-device pair as illustrated in Fig. \ref{AI_downloading_system}. 
The edge server transmits the well-pretrained AI models to the local device for execution of user-specific inference tasks. 
Relevant system model and metrics are presented as follows.
\vspace{-3mm}
\subsection{AI Downloading System}
\begin{figure*}[t!]
\centering
\includegraphics[width=0.5\textwidth]{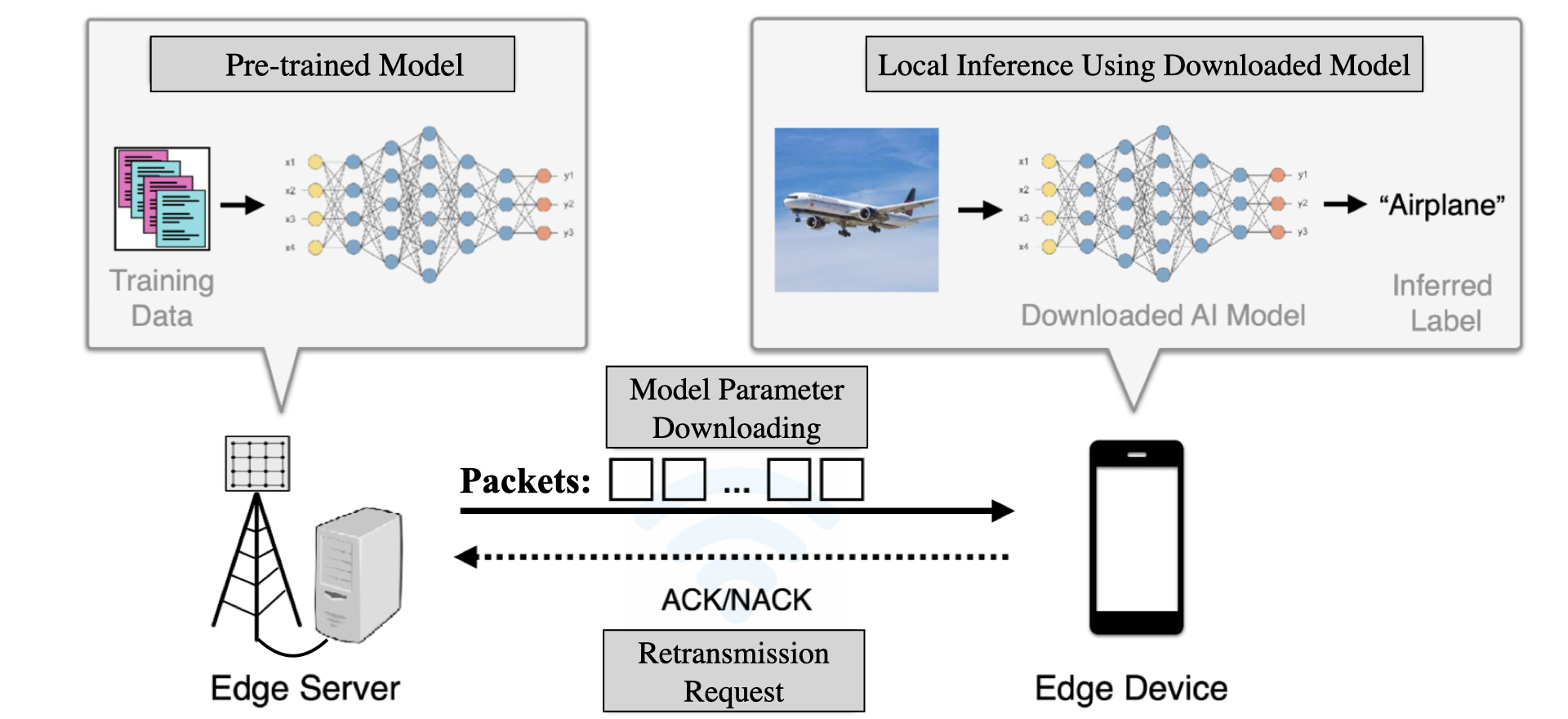}\vspace{-1mm}
\caption{The AI-model downloading system with retransmission.}
\label{AI_downloading_system}\vspace{-6mm}
\end{figure*}

\subsubsection{Packetization}
At the edge server side, the AI model's parameters are partitioned and encapsulated into multiple packets. 
We utilize modulation and coding scheme (MCS) tables to select suitable combinations of modulation order and code rate for downlink transmission. 
According to practical wireless communication standards such as LTE and 5G NR,
the transmitter and receiver typically agree on an MCS index, which corresponds to a particular modulation order (e.g., QPSK, 16-QAM, 64-QAM, 256-QAM, etc.) and associated code rate (e.g., 1/2, 3/4, 2/3, 5/6, etc.). 
Following the same argument for conventional ARQ/HARQ protocols, we assume the transmitter has no instantaneous channel state information (CSI) when making (re)transmission decisions. 
Consequently, the modulation order and code rate are chosen based on long-term (historical) channel statistics rather than per-packet CSI. 
The information-bearing payload per packet is determined by the modulation order $M$, which provides $\log_2 M$ bits per symbol, and the code rate $\alpha\in(0,1]$, which is the fraction of information bits after channel coding. 
If each packet carries $W$ modulation symbols and each parameter is encoded using $n$ bits, the maximum number of parameters that can be conveyed in one packet is
\begin{align}
    K_\text{capacity} = \left\lfloor \frac{\alpha W\log_2 M}{n} \right\rfloor,
\end{align}
where $\lfloor\cdot\rfloor$ denotes the integer floor function.
\subsubsection{Parameter Transmission}
Consider the downlink transmission of a neural network model whose parameters are concatenated into a $D$-dimensional vector $\mathbf{w}\in\mathbb{R}^{D}$. 
Let $w_d$ denote the $d$-th (real-valued) parameter, which is quantized and encoded as an $n$-bit signed integer \cite{jacob2018quantization}. The represented decimal integer value is expressed as
\begin{equation}
    w_d=\sum_{i=0}^{n-2}b_i2^i-b_{n-1}2^{n-1},
\end{equation}
where $b_i\in\{0,1\}$ is the binary value at the $i$-th position. 
Each bit is transmitted over a binary symmetric channel with error probability $P_b$. Let $a_i\in\{0,1\}$ be independent flip indicators. The received bit is $\hat{b}_i=b_i\oplus a_i$, and the resulting distorted integer is~\cite{zeng2024knowledge}:
% \begin{equation}
%     \hat{w}_d=\sum_{i=0}^{n-2}(b_i\oplus a_i)2^i-(b_{n-1}\oplus a_{n-1})2^{n-1}.
% \end{equation}
% Equivalently, in algebraic form, 
\begin{equation}\label{eqn: received feature}
    \hat{w}_d=\sum_{i=0}^{n-2}[a_i(-1)^{b_i}+b_i]2^i-[a_{n-1}(-1)^{b_{n-1}}+b_{n-1}]2^{n-1},
\end{equation}
where the binary indicators $\{a_i\}$ i.i.d. follow $\text{Bernoulli}(P_b)$, independent across bit positions and across parameters.

\vspace{-3mm}
\subsection{Performance Metrics}

\subsubsection{Downloading Loss}
In neural networks, the inference accuracy is commonly assessed via the validation loss $f(\mathbf{w})$, defined as a function of the model’s parameter vector $\mathbf{w}\in\mathbb{R}^D$ that collects all trainable weights. 
Typically, a lower validation loss corresponds to higher inference accuracy, indicating an inverse relationship between these metrics \cite{yan2022swim}. 
We refer to a model as well trained if its parameters achieve (or closely approximate) the minimum validation loss on the validation set.
Given the close relationship between accuracy and loss, we quantify each parameter’s importance by its impact on the loss. Let $\mathbf{w}$ denote the pre-trained optimal parameters. Perturbation (e.g., bit errors) of $\mathbf{w}$ induces an increase in the loss function, which we refer to as the \emph{loss variation}:
\begin{equation}\label{lossvariation}
    \Delta f(\mathbf{w}) = f(\hat{\mathbf{w}}) - f({{\mathbf{w}}}).
\end{equation}
We denote the error in the parameter vector $\mathbf{w}$ as the random vector $\Delta {\bf w}=\hat{\bf w}-\mathbf{w}$. Using a second-order Taylor expansion around $\mathbf{w}$, the loss variation is written as
\begin{align}\label{taylor}
    \Delta f(\mathbf{w}) \approx \nabla f(\mathbf{w})^T\!\Delta{\mathbf{w}} + \frac{1}{2}\Delta \mathbf{w}^T \mathcal{H}({\mathbf{w}})\Delta \mathbf{w} + o(\Delta {\mathbf{w}}^3),
\end{align}
where $\mathcal{H}({\mathbf{w}})=\nabla^2f(\mathbf{w})$ denotes the Hessian matrix of ${\mathbf{w}}$. 
To quantify how these perturbations affect model performance, we consider the expected  loss variation, $\mathbb{E}[\Delta f(\mathbf{w})]$, induced by distortions in model parameters\cite{yan2024hessian}:
\begin{align}\label{hu's paper}
    \mathbb{E}\left[\Delta f(\mathbf{w})\right] \nonumber&\overset{(a)}{\approx} \frac{1}{2}\mathbb{E}[\Delta \mathbf{w}^T \mathcal{H}\!\left({\mathbf{w}}\right)\Delta \mathbf{w}] \nonumber\\&\overset{(b)}{\approx} \frac{1}{2}\sum_{d=1}^{D}\frac{\partial^2 f}{\partial w_d^2}\cdot \mathbb{E}\left[\left(\Delta{w}_d\right)^2\right],
\end{align}
where the approximation (a) comes from the Taylor expansion together with the fact: $\mathbb{E}[\Delta {w}_d]=0, \forall d\in\{1,\cdots,D\}$; the approximation (b) arises from the fact that the off-diagonal elements of the Hessian matrix contribute negligibly to the expected loss variation~\cite{yan2024hessian}. We define each parameter's second-order derivative, $\frac{\partial^2 f}{\partial w_d^2}$  (i.e., the diagonal Hessian element $\mathcal{H}_{d,d}$), as \textit{parametric sensitivity}: it quantifies the expected contribution of perturbations in $w_d$ to the overall loss. Larger values indicate that distortions in $w_d$ result in proportionally greater degradation in model performance, making such parameters more critical to preserve during noise injection\footnote{Since the model is well trained, the Hessian matrix is positive semi-definite at a local minimum as: $\nabla^2 f(\mathbf{w}) \succcurlyeq 0.$
Thus, each second derivative of the loss function with respect to parameter $w_d$ is non-negative: $\frac{\partial^2 f}{\partial w_d^2} \geq 0, \forall d.$}.
Consider an arbitrary packet $\mathcal{U}_j$ in the AI downloading system; we define its \emph{packet-level sensitivity} as the cumulative sensitivity of all parameters contained within the packet:
\begin{equation}
    s_j \triangleq \sum_{d \in \mathcal{U}_j} \frac{\partial^2 f}{\partial w_d^2}.\label{Eqn: 7}
\end{equation}
In the AI downloading system, we define the \emph{downloading loss} as the expected loss variation $\mathbb{E}[\Delta f(\mathbf{w})]$ resulting from channel-induced noise that perturbs the model parameters during transmission. The detailed analysis of this metric is presented in the Section III-A.

\subsubsection{Downloading Latency}
The downloading latency is quantified as the total number of transmissions of total packets.
% Given a total of $J$ packets, each with a fixed transmission size of $W$ bits, modulation order $M$, the total number of transmission symbols required to transmit the AI-model once can be expressed as:
% \begin{align}
%     N_\text{symbol} = \left\lceil \frac{JW}{\log_2 M} \right\rceil,
% \end{align}
% where $\sum_{i=1}^{J} k_i = D$.
% The total number of retransmission slots is obtained by tracking each packet independently. 
Consider total $J$ packets for transmission. 
If the $j$-th packet's received BER exceeds its prescribed reliability threshold, the packet is retransmitted and its per-packet counter $T_j$ is incremented. 
This is repeated until the packet meets its reliability requirement. 
Summing the counters across all packets gives the total number of transmission slots:
\begin{equation}
    T_{\rm total}=\sum_{j=1}^{J}T_j,
\end{equation}
which reflects the communication latency required to achieve the target model downloading reliability.
\vspace{-3mm}
\section{Sensitivity-Aware Downloading Loss}
In this section, we develop a sensitivity-aware downloading loss metric that quantifies the contribution of individual parameters to downloaded model's loss variation in the presence of channel distortions. Building on this metric, we characterize the heterogeneous importance of model parameters by analyzing their sensitivity distributions in representative neural network architectures. Moreover, we reveal the inherently right-skewed sensitivity distributions of neural network parameters and demonstrate how sensitivity skewness leads to highly unequal impacts of parameters on model performance. These results lay a theoretical foundation for sensitivity-aware retransmission design in the following sections.
\begin{figure}[t!]
\centering
\subfigure[LeNet-5]{
\includegraphics[width=0.6\columnwidth]{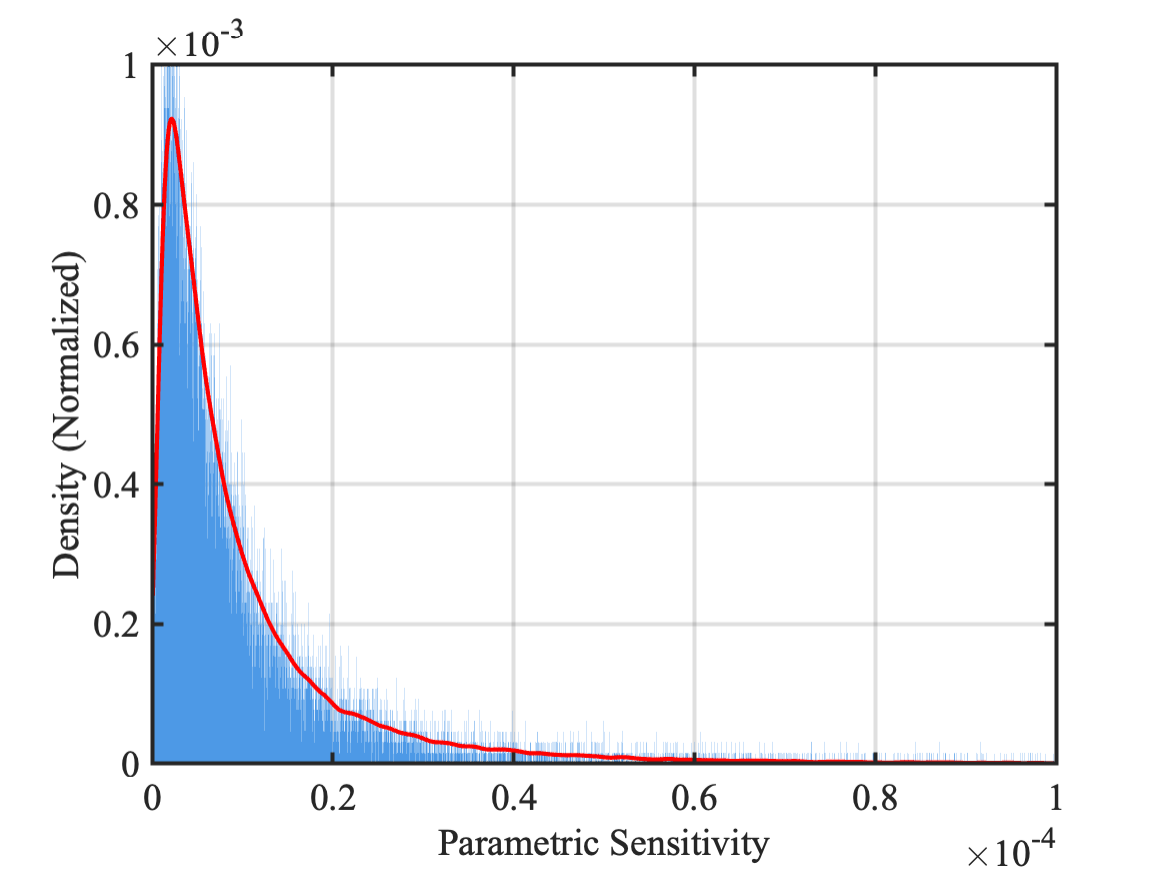}\label{skewness_LeNet}}\vspace{-1mm}
\subfigure[ShuffleNetV2]{
\includegraphics[width=0.6\columnwidth]{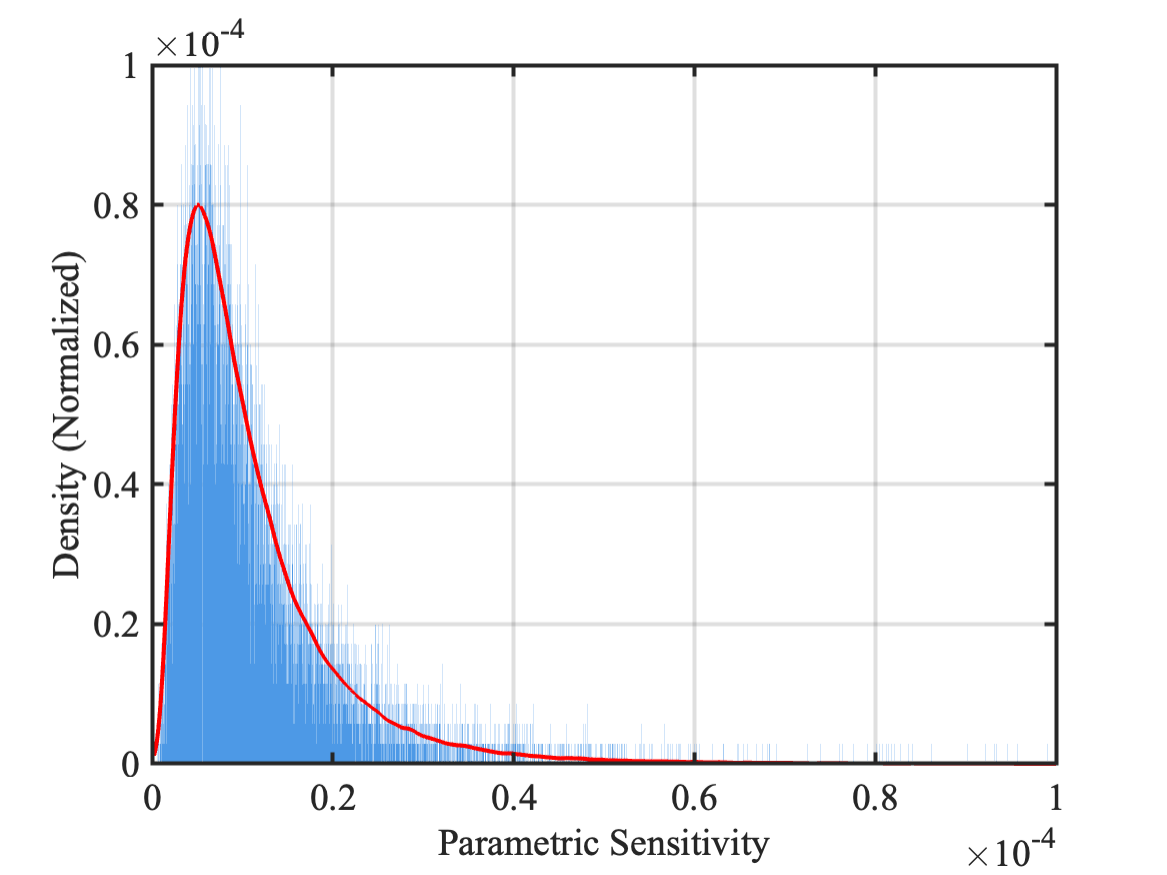}\label{skewness_ShuffleNet}}\vspace{-1mm}
\caption{Distribution of the parametric sensitivity in two DNN models. Each parametric sensitivity is assigned to a histogram bin, and the bin counts are normalized to form a probability density function. To obtain a smooth approximation of the distribution, every 20 consecutive bins are grouped, and the average bin center and corresponding average density are computed for each group. The red curve connects these averaged points to represent the underlying sensitivity distribution.}
\label{skewness_ditribution}\vspace{-6mm}
\end{figure}
\vspace{-5mm}
\subsection{Downloading Loss Analysis}
Based on \eqref{hu's paper}, we introduce below the \emph{sensitivity-aware downloading loss} to quantify the expected loss variation resulting from bit-error-induced parametric distortion.

\begin{Lemma}[Sensitivity-Aware Downloading Loss]\label{Lemma0}\rm
Consider a model with $J$ packets, where the parameters in each packet are encoded as $n$-bit signed integers and transmitted over a wireless channel. Each packet experiences a potentially different BER, denoted by $P_{b,j}$. 
Accordingly, the expected sensitivity-aware downloading loss of the model is given by\vspace{-2mm}
\begin{align} \label{importancescore}
\mathbb{E}\left[\Delta f(\mathbf{w})\right] \approx \alpha\cdot\sum_{j=1}^{J}s_j\cdot P_{b,j},
\end{align}
where $\alpha=\frac{4^n-1}{6}$ denotes the constant term under a fixed $n$-bit quantization.
\end{Lemma}
\begin{proof}
    (See Appendix~\ref{Lemma0proof}).
\end{proof}
% \noindent The proof is given in Appendix \ref{Lemma0proof}.

One can observe that the downloading loss is governed by two main factors:
\begin{itemize}
    \item \textit{Parametric sensitivity:} 
    The packet-level parametric sensitivity $s_j$, defined as the sum of the parametric sensitivities with respect to the parameters contained in packet $\mathcal{U}_j$, scales the impact of a parameter distortion. For a fixed noise level, packets with higher cumulative sensitivity contribute more significantly to model performance degradation.
    
    \item \textit{Channel-induced distortion:}  
    The term $P_{b,j}$ captures the channel-induced distortions that affect the parameters contained in packet $\mathcal{U}_j$. 
    As $P_{b,j}$ increases due to poorer channel conditions, the induced distortion grows, leading to a larger contribution of $\mathcal{U}_j$ to the downloading loss.
\end{itemize}
Within a given packet, the channel-induced distortion affects all parameters in one packet uniformly, whereas parametric sensitivity varies across parameters.
Since this sensitivity governs each parameter’s contribution to model's downloading loss, characterizing its distribution across parameters is essential for understanding the heterogeneous impact on model performance.
\begin{figure}[t!]
\centering
\subfigure[Average downloading loss.]{
\includegraphics[width=0.6\columnwidth]{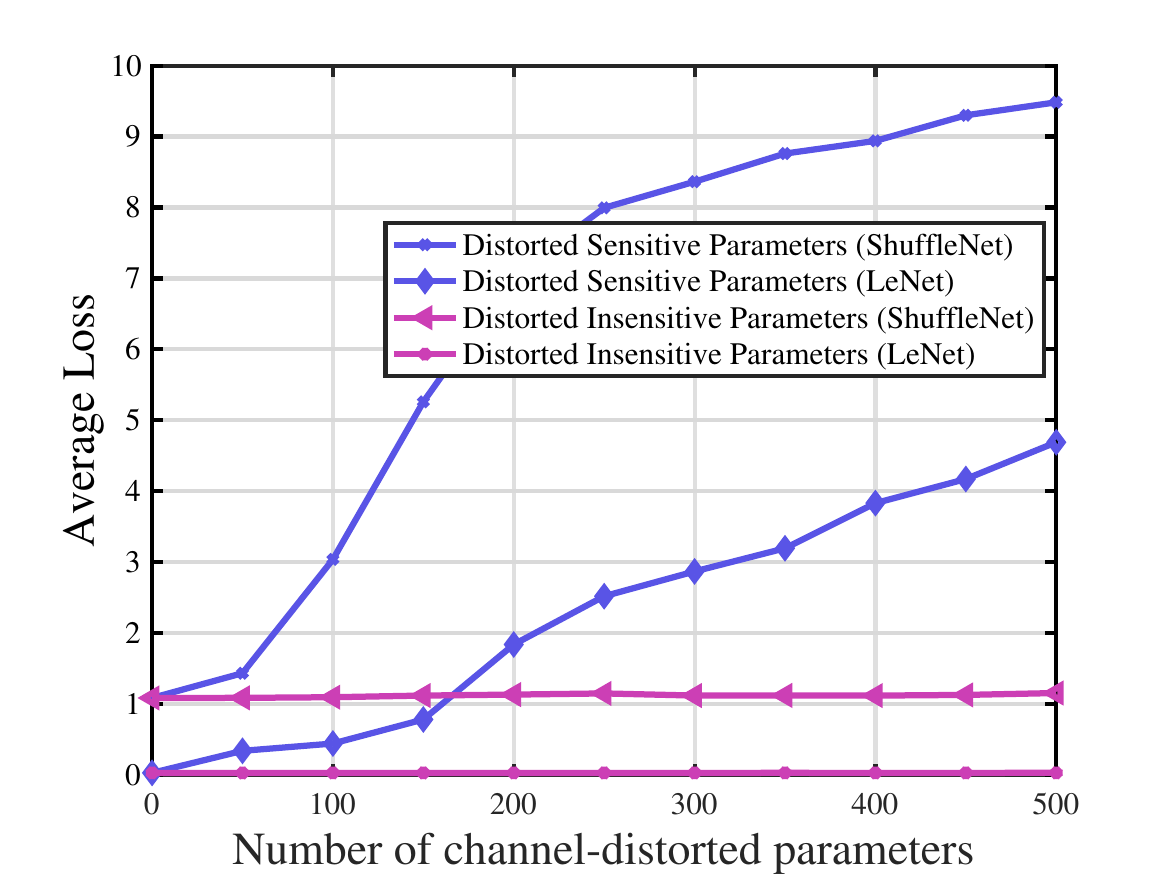}\label{LeNet_error}}\vspace{-1mm}
\subfigure[Average classification accuracy.]{
\includegraphics[width=0.6\columnwidth]{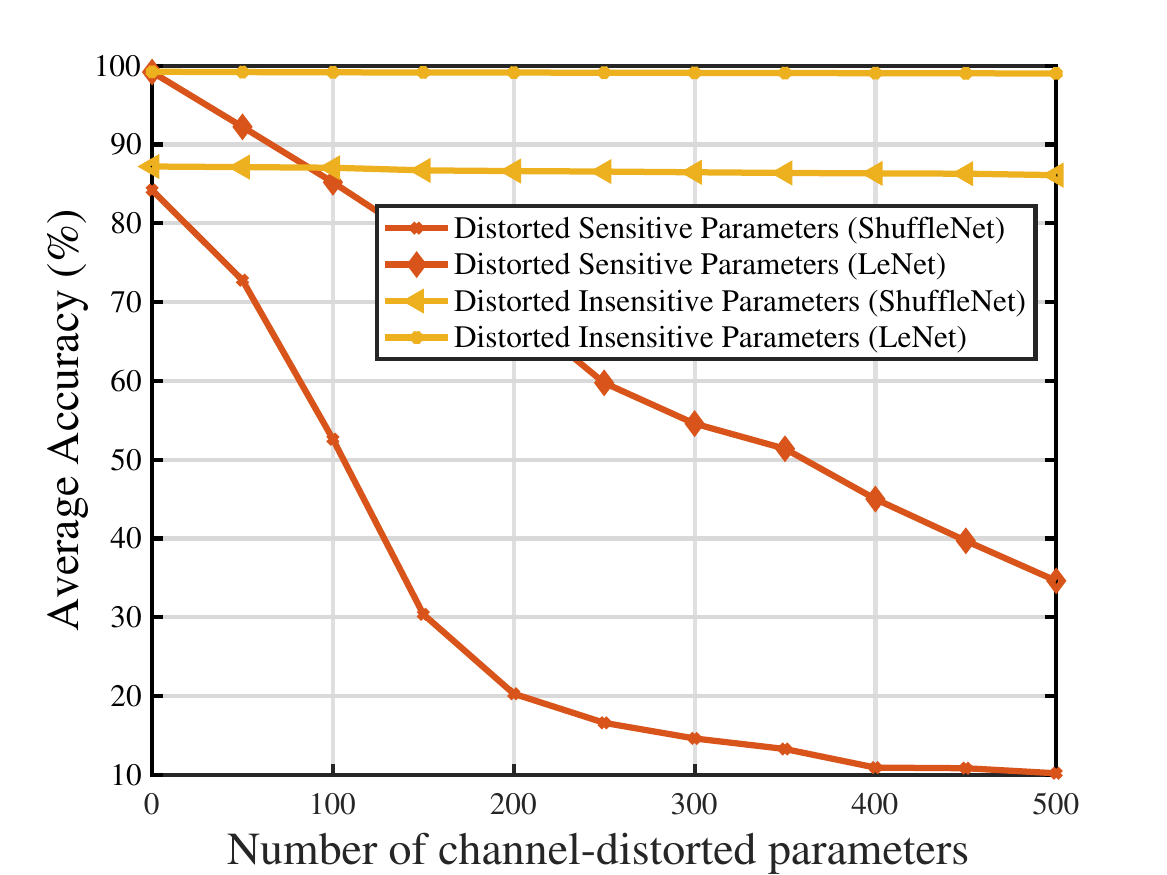}\label{Shufflenet_error}}\vspace{-1mm}
\caption{The effect of injecting BER = 0.1 into high- versus low-sensitivity parameter subsets (top-500 vs. bottom-500) on both model average loss and inference accuracy.}
\label{errortolerance}\vspace{-6mm}
\end{figure}\vspace{-3mm}
\subsection{Sensitivity Skewness and Downloading Loss}\label{Sec: skewness}
Unlike generic data that is typically treated as uniformly important, deep neural networks exhibit substantial variability in the importance of individual parameters. 
To quantitatively demonstrate such inherent within-architecture heterogeneity, the parametric-sensitivity distributions for LeNet-5 and ShuffleNetV2 are shown in Fig.~\ref{skewness_ditribution}. We formally quantify their statistical characteristics by the skewness measure as follows.
\begin{Lemma}[Skewness Measure \cite{skewnessmedianmean}]\rm
    The skewness of the probability distribution of a random variable $X$ can be characterized by the Fisher's moment coefficient, given by
\begin{align}
     \gamma_X
 = \frac{\mathbb{E}[(X - \mathbb{E}[X])^3]}
        {(\operatorname{Var}[X])^{3/2}}.
\end{align}
\end{Lemma}

\begin{Remark}[Skewness of Parametric Sensitivity]\rm
    The skewnesses of LeNet-5 and ShuffleNetV2 are calculated as 2.36 and 1.89, respectively. These values satisfy the $|\gamma_X|\ge1$ criterion for highly right-skewed behavior~\cite{highlyskewed}. This phenomenon indicates that most parameters have low sensitivities, while only a small subset exhibits high sensitivity and consequently dominates the model's downloading loss under channel distortions.
\end{Remark}

To demonstrate the inefficiency of uniform-reliability transmission, we compare the effect of bit errors on a model's most and least sensitive parameters.
For each model (ShuffleNetV2 and LeNet-5), we rank parameters by parametric sensitivity and form two groups: one with the 500 parameters exhibiting the highest sensitivities and the other with the 500 parameters exhibiting the lowest sensitivities. We then inject bit errors only into the selected group at a BER of 0.1, leaving all other parameters unchanged. As shown in Fig. \ref{errortolerance}, distorting the high-sensitivity group causes a significant drop in accuracy, whereas corrupting the low-sensitivity group has a negligible effect. This disparity highlights the inefficiency of uniform transmission policies and lays the foundation for sensitivity-based retransmission design in the sequel.

\vspace{-3mm}
\section{Overview of PASAR Protocol}
In this section, we design the PASAR protocol by adopting the solution of the online multiple-choice knapsack problem (MCKP). To this end, the retransmission control problem is formulated to minimize the total transmission latency under a constraint on the total sensitivity-aware downloading loss. By solving the problem as mentioned, the PASAR protocol is designed with the key feature of an adaptively adjusted termination threshold. The detailed threshold design is the topic of the next section.
\vspace{-3mm}
\subsection{Retransmission Control Problem}
First, the retransmission process is described as follows. In AI model downloading system, the model parameters are partitioned and encapsulated into $J$ packets, with the $j$-th packet containing a partial vector $\mathbf{w}_j$ (see Section II-A). Each packet $\mathcal U_j$ is assigned a packet-level sensitivity $s_j$ as defined in \eqref{Eqn: 7}.
Following retransmission mechanism, a packet is retransmitted iteratively until it satisfies the termination criterion.
Upon the $i$-th reception of packet $\mathcal U_j$, the device decodes it to obtain a partial parameter vector ${\mathbf w}_{j,i}$ and record an instantaneous BER $P_{b,j,i}$.
Given packet $\mathcal U_j$ terminates after $T_j$ transmissions, the device averages its $T_j$ decoded parameters to form $\widehat{\mathbf w}_j=\frac{1}{T_j}\sum_{i=1}^{T_j}{\mathbf w}_{j,i}$, which reduces the variance of channel-induced distortion by $1/T_j$. The inference task is then performed using the downloaded model assembled from $\{\widehat{\mathbf w}_j\}_{j=1}^J$.

Next, we quantify the total downloading loss over rounds.
For an arbitrary packet $\mathcal U_j$, let $T_j$ denote the number of transmissions required to meet its stopping criterion.
Due to channel fading, each transmission experiences a different BER at the receiver, denoted as $\{P_{b,j,i} : i = 1, \ldots, T_j\}$. 
For each parameter $d\in\mathcal{U}_j$, the model's downloading loss caused by this parameter after $(T_j-1)$ retransmissions is given by:
\begin{align}\label{parameterloss}
    \mathbb{E}\left[\Delta f(w_d)\right] \approx \alpha\cdot\frac{\partial^2 f}{\partial w_d^2} \cdot \overline{P}_{b,j,T_j},
\end{align}
where the receiver's average BER of the packet $\mathcal{U}_j$ after $(T_j-1)$ retransmissions is denoted as
\begin{equation}
    \overline{P}_{b,j,T_j}\triangleq\frac{1}{T_j}\sum_{i=1}^{T_j} P_{b,j,i}.
\end{equation}
By summing over the downloading loss caused by all parameters in the packet $\mathcal{U}_j$, the accumulated downloading loss for this packet after the $T_j$-th transmission is
\begin{align}\label{packetloss}
\beta_{j,T_j}&\approx \alpha\cdot \sum_{d \in \mathcal{U}_j} \frac{\partial^2 f}{\partial w_d^2} \cdot \overline{P}_{b,j,T_j}\nonumber\\&=\alpha\cdot s_j\cdot \overline{P}_{b,j,T_j}.
\end{align}
Consider a model whose parameters are partitioned into $J$ packets for transmission, where each packet $\mathcal{U}_j$ contributes the corresponding loss $\beta_{j,T_j}$. To ensure that the cumulative performance degradation remains within an acceptable range, the total downloading loss is subject to a global constraint $\beta_{\mathrm{total}}$ that characterizes the performance of the assembled complete model on the device's side.

Finally, we formulate the retransmission control problem with the objective of minimizing the total retransmission latency while satisfying the constraint on downloading loss. Mathematically,
\begin{equation*}
\textbf{(P1)}\qquad
\begin{aligned}
\min_{\{T_j\}} \quad& \sum_{j=1}^J T_j \\
\text{s.t.}\quad& \sum_{j=1}^J \beta_{j,T_j} \le \beta_{\mathrm{total}}.
\end{aligned}
\end{equation*}
\vspace{-8mm}
\subsection{MCKP Approach}
    Problem P1 exhibits the structural characteristics of an MCKP \cite{MCKP}. Specifically, for each packet $\mathcal{U}_j$, exactly one stopping round $T_j$ must be selected, where selecting round $T_j$ incurs an associated loss budget $\beta_{j,T_j}$. The traditional MCKP is an offline combinatorial optimization problem in which all item attributes are fixed and known prior to decision making. Even in this relatively simple setting, the MCKP is NP-hard, meaning that exact global optimization is computationally intractable when $J$ is large\cite{MCKP}. 
    Solving problem P1 to support online implementation is more challenging due to the causality constraint of channel information (i.e., information available only up to the current round).
    As a result, online decisions must be made based on historical information without knowledge of future channel realizations. The traditional round-based thresholding structure is commonly adopted in the literature to tackle online MCKP by decomposing the sequential decision-making process into discrete rounds with an adaptively adjusted termination threshold~\cite{MCKP2,MCKP3}. The resultant protocols prioritize packets with high value-to-cost ratios as the budget diminishes and are theoretically shown to achieve provable performance guarantees~\cite{MCKP2,MCKP3}. However, existing schemes typically control retransmissions solely based on channel conditions, treating all packets uniformly and ignoring the heterogeneous parametric sensitivities that govern their unequal contributions to the model performance. To enhance the system performance, we propose the MCKP-type PASAR protocol in the next sub-section with the distinctive feature of sensitivity-aware thresholding framework.
\begin{figure*}[t!]
\centering
\includegraphics[width=0.52\textwidth]{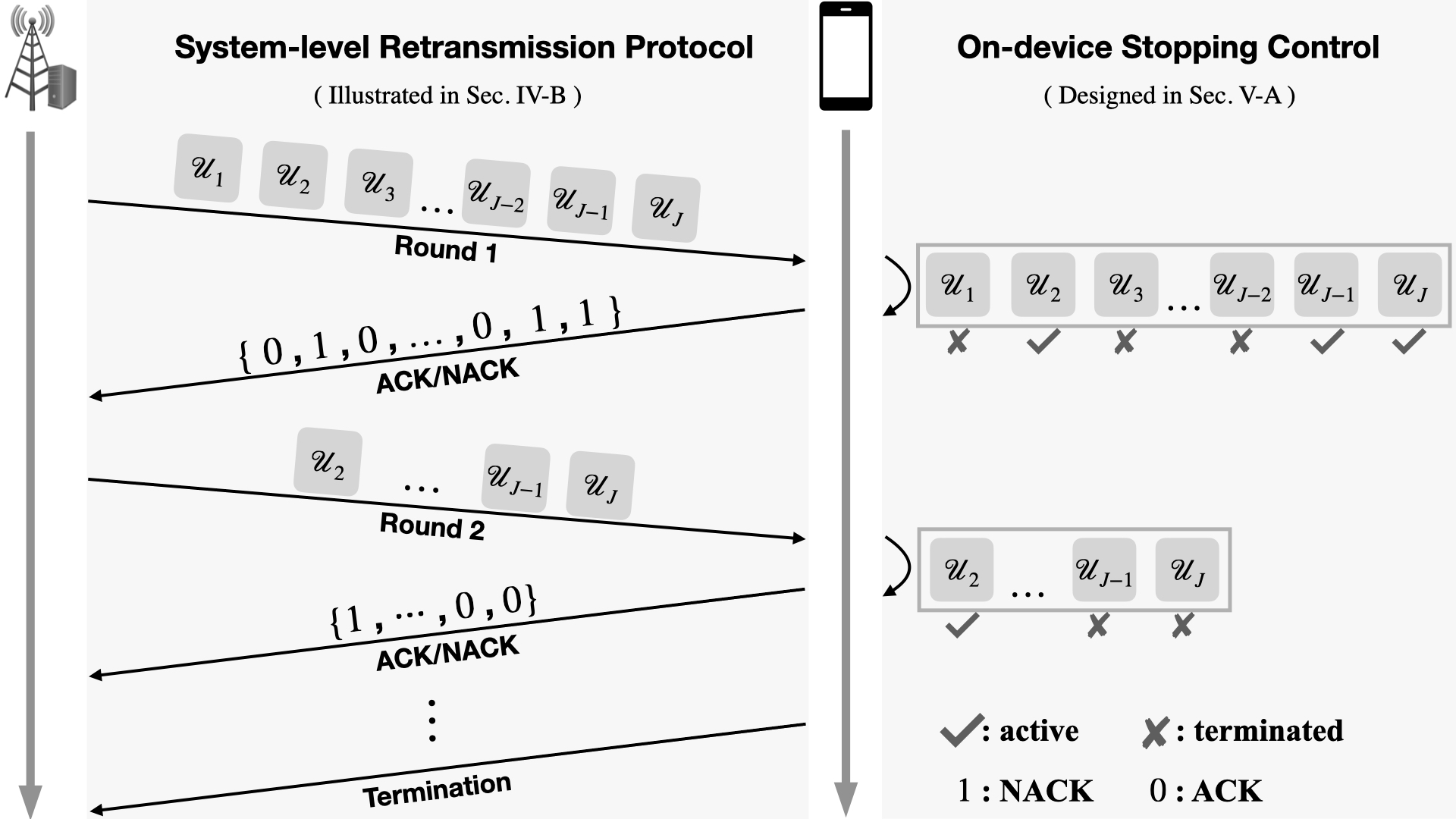}\vspace{-2mm}
\caption{Online retransmission control of the PASAR protocol.}
\label{PASAR_protocol_workflow}\vspace{-5mm}
\end{figure*}
\vspace{-3mm}
\subsection{PASAR Protocol}\label{operations} 
Motivated by the round-based thresholding framework for online MCKP, we tailor it to the model-downloading setting by incorporating packet-level parametric sensitivity, yielding the proposed PASAR protocol.
From an operational perspective, PASAR differs from conventional retransmission protocols in how retransmission decisions are executed during the downloading process.
Rather than relying on a fixed or globally uniform BER threshold, PASAR controls packet retransmission using packet-dependent BER criteria (see Step 3 below). Such criteria are determined by the sensitivity of the parameters carried by each packet.
As a result, packets containing less sensitive parameters can be terminated earlier even under relatively high BERs, while packets carrying highly sensitive parameters are retained for further retransmission to ensure sufficient reliability. In other words, the communication budget is better utilized by prioritizing reliable delivery of more critical parameters.

As shown in Fig.~\ref{PASAR_protocol_workflow}, the PASAR protocol adopts a two-level architecture consisting of a system-level retransmission framework and an on-device stopping control module.  
Its key steps are elaborated below.

\begin{itemize}
    \item Step 1 \emph{(Initialization):}
The protocol initiates by computing a look-up table of dimension $D$, among which the $d$-th entry corresponds to the parametric sensitivity, $\frac{\partial^2 f}{\partial w_d^2}$, for the $d$-th parameter. This table serves as the basis for evaluating the sensitivity of each packet in the subsequent phase. The parametric sensitivity table together with channel encoding and modulation operations provide the packet sensitivities $\{s_j\}$ and the initial packetization $\{\mathcal U_j\}_{j=1}^J$, which remain fixed throughout the protocol. Let $\beta_{\text{res},t}$ denote the remaining loss budget available in the transmission round $t$, and let $\mathcal V_{t}$ represent the set of active packets in this round. For round $t=1$, the initial residual budget is set to \( \beta_{\mathrm{res},1} = \beta_{\mathrm{total}} \), and all packets are active: \( \mathcal{V}_1 = \{\mathcal{U}_1, \ldots, \mathcal{U}_J\} \).

    \item Step 2 \emph{(Packet transmission):} 
In each round \( t \), the server transmits the current active packet set \( \mathcal{V}_t \) to the edge device for evaluation. Encoded packets are transmitted sequentially based on the selected MCS. The values $\{s_j\}_{j=1}^J$ are embedded as metadata in the header of each packet for reliable delivery to the device. After transmission, the device receives the symbols and demodulates them into a bit stream. The demodulated bit stream is then passed through the channel decoder and mapped back to the corresponding model parameters. In this case, the BER of packet $\mathcal{U}_j\in \mathcal{V}_t$ in its $i$-th transmission, denoted $P_{b,j,i}$, is treated as a known value obtained from the channel state information at the receiver (CSIR). 

    \item Step 3 \emph{(Retransmission control):}  
This is the module that represents the novelty of the proposed PASAR protocol. For each packet, the receiver updates the cumulative average BER after each retransmission and compares it with a packet-specific stopping threshold. Here, we define that a packet $\mathcal U_j$ is \emph{terminated} in round $t$ if the device terminates its retransmission process for $\mathcal U_j$ within this round, indicating that the packet has been successfully delivered. Packets that meet the stopping criterion are terminated and marked as ACK, while the remaining packets are labeled as NACK. The device then sends these ACK/NACK indicators to the server, which aggregates the NACK feedback to update the active set for the next round, i.e., $\mathcal{V}_{t+1} \leftarrow \{\mathcal{U}_j \in \mathcal{V}_t : \text{NACK received}\}$. The detailed design is presented in the next section.

 \item Step 4 \emph{(Protocol stopping criterion):} The PASAR protocol ends when one of the following conditions is met: 1) All packets are terminated within the allowed maximum number of transmissions, i.e., $\mathcal{V}_t = \emptyset$ for $\sum_{i=1}^t|\mathcal{V}_i| \le T_{\mathrm{max}}$, where $|\cdot|$ denotes the cardinality of a set. 2) The transmission limit is exceeded, i.e., $\sum_{i=1}^t|\mathcal{V}_i| > T_{\mathrm{max}}$ while $\mathcal{V}_t \ne \emptyset$.
\end{itemize}

The operational flow of the PASAR protocol is summarized in Algorithm~\ref{alg:SystemPASAR}. 

\begin{algorithm}[t]
\caption{PASAR Protocol}
\label{alg:SystemPASAR}
\begin{algorithmic}[1]
\STATE \textbf{Input:} Total budget $\beta_{\text{total}}$, max round $T_{\max}$, initial packet set $\mathcal{V}_{t=1} = \{\mathcal{U}_1, \ldots, \mathcal{U}_J\}$;
\STATE \textbf{Initialize:} $\beta_{\text{res},1} \leftarrow \beta_{\text{total}}$, $t \leftarrow 1$;
\WHILE{$\mathcal{V}_t \neq \emptyset$ \AND $\sum_{i=1}^t|\mathcal{V}_i| \le T_{\max}$}
    \STATE Server transmits $\mathcal{V}_t$ to device;
    \STATE Local device executes \textbf{On-Device Stopping Control (Algorithm~\ref{alg:OnDevice})}
    \STATE Device sends per-packet feedback (ACK/NACK) for all packets in $\mathcal{V}_t$;
    \STATE Server updates the active set for the next round as $\mathcal{V}_{t+1} \leftarrow \{\mathcal{U}_j \in \mathcal{V}_t : \text{NACK received}\}$;
    \STATE $t \leftarrow t + 1$;
    \ENDWHILE
\IF{$\mathcal{V}_t \neq \emptyset$}
    \RETURN Failure: Time/budget exceeded;
\ENDIF
\end{algorithmic}
\end{algorithm}

\vspace{-3mm}
\section{Sensitivity-aware Retransmission Control}
In this section, we design the key module (step) in the PASAR: retransmission control. At the core of the associated algorithm, the computation of dedicated BER thresholds incorporates the sensitivity metrics for received packets at the device side. We prove the effectiveness of the greedy method used to design the algorithm and then compare the design to conventional channel-aware ARQ/HARQ protocols. In addition, the efficiency of the threshold computation is shown via complexity analysis.

\vspace{-3mm}
\subsection{Design of Control Threshold}\label{Sensitivity Aware Retransmission Control}
In retransmission control module (Section~IV-B, Step~3), PASAR adopts a round-wise greedy principle: in each retransmission round, it makes non-anticipative decisions based only on the currently observed empirical BERs and packet-level sensitivities, and seeks to terminate as many packets as possible to shrink the active set early and reduce subsequent retransmissions. 
Direct implementation requires sorting $J$ packets globally by their instantaneous loss-budget consumption [costing $\mathcal{O}(J\log J)$], and then sequentially scanning the sorted list with per-packet threshold comparisons [costing an additional $\mathcal{O}(J)$]. This per-round complexity can be too high for real-time on-device execution. We instead develop a sensitivity-aware controller that implements the same greedy termination with a much lower complexity (see Section V-B). Specifically, it replaces global sorting with parallel, per-packet threshold checks calibrated by packet-level sensitivity, followed by a greedy budget-filling refinement.
Given the highly right-skewed empirical distribution of parametric sensitivities (see Section~\ref{Sec: skewness}), our design enables rapid elimination of the least sensitive packets while reserving retransmission resources for the small subset of highly sensitive ones.

The detailed retransmission-control algorithm is presented as follows. Upon receiving retransmissions in round $t$, the device evaluates whether each active packet $\mathcal U_j \in \mathcal V_t$ should terminate or remain in the active set, based on its packet-level sensitivity $s_j$ and empirical average BER $\overline{P}_{b,j,t}$. We introduce an index $\ell\in\{1,2,\cdots\}$ that counts the number of threshold-update epochs for each transmission round. Correspondingly, two quantities evolve with $\ell$: the residual loss budget $\beta_{\mathrm{res},t,\ell}$ and the active packet set $\mathcal V_{t,\ell}$. We design an adaptive sensitivity-aware termination threshold that is updated within each round based on the residual loss budget and the current active set. The algorithm comprises two phases:
\begin{itemize}
    \item Phase 1 \emph{(Threshold-based filtering):} The following two procedures are repeated until no remaining packet satisfies \eqref{termination_con}.
    \begin{itemize}
    \item[(i)] \emph{Termination threshold:} At update epoch $\ell$ in round $t$, we design a sensitivity-scaled termination threshold for each active packet $\mathcal U_j\in\mathcal V_{t,\ell}$ as
\begin{equation}\label{termination_con}
        \Gamma_{j,t,\ell} \triangleq  \frac{\beta_{\mathrm{res},t,\ell}}{\alpha\cdot s_j \cdot |\mathcal V_{t,\ell}|}.
    \end{equation}
    $\mathcal U_j$ is terminated if its average BER $\overline{P}_{b,j,t}= \frac{1}{t}\sum_{i=1}^{t}P_{b,j,i}$ satisfies $\overline{P}_{b,j,t}\le \Gamma_{j,t,\ell}$.
    \item[(ii)]\emph{Threshold updating:} After all packets satisfying the termination criterion at update iteration $\ell$, the residual loss budget and active set are updated by removing the terminated packets and deducting their corresponding loss contributions: 
\begin{equation}\label{update1}
        \beta_{\mathrm{res},t,\ell+1}
    =\beta_{\mathrm{res},t,\ell}
    -\!\!\!\!\!\sum_{\mathcal U_j\in \mathcal V_{t,\ell}:\overline{P}_{b,j,t}\le \Gamma_{j,t,\ell}}
    \!\!\!\!\!\!\!\!\!\alpha\cdot s_j \cdot\,\overline{P}_{b,j,t},
\end{equation}
\begin{equation}\label{update2}
        \mathcal V_{t,\ell+1}
    =\mathcal V_{t,\ell}\setminus 
    \left\{\mathcal U_j\in \mathcal V_{t,\ell}:\overline{P}_{b,j,t}\le \Gamma_{j,t,\ell}\right\}.
\end{equation}
The threshold-update index is incremented as $\ell \leftarrow \ell + 1$, and the corresponding threshold in \eqref{termination_con} is recalculated based on the updated active set $\mathcal V_{t,\ell}$ and loss budget $\beta_{\mathrm{res},t,\ell}$.
\end{itemize}
    \item Phase 2 \emph{(Greedy refinement):}
    Let $\mathcal V'_{t}$ and $\beta'_{\mathrm{res},t}$ denote the resulting active set and residual budget, respectively. The device then greedily exploits any remaining loss budget to terminate additional packets with the smallest loss contributions. Specifically, the device selects a subset $\mathcal S_t \subseteq \mathcal V'_t$ with the smallest $ \alpha s_j \overline{P}_{b,j,t}$ such that $\sum_{\mathcal U_j\in\mathcal S_t} \alpha s_j \overline{P}_{b,j,t} \le \beta'_{\mathrm{res},t}$.
    The selected packets are declared terminated and removed from the active set. 
\end{itemize}
At this point, the device locally stores the updated residual loss budget (i.e., $\beta_{\text{res},t+1}=\beta'_{\text{res},t}$) and reports per-packet ACK/NACK feedback to the server. The online implementation of the proposed on-device retransmission control mechanism is summarized in Algorithm~\ref{alg:OnDevice}. 

\begin{Lemma}[Greedy Property of the Threshold Design]\rm\label{Lemma3}
Fix an arbitrary round $t$ and let $\mathcal V_t$ denote the set of active packets. 
By selecting the termination set $\mathcal S_t \subseteq \mathcal V_t$, the proposed adaptive termination-threshold scheme maximizes the number of successfully delivered packets under the residual loss-budget constraint:
\[
|\mathcal S_t|=\max_{\mathcal S\subseteq \mathcal V_t:\ \sum_{\mathcal U_j\in\mathcal S} \alpha s_j \overline{P}_{b,j,t}\le \beta_{\mathrm{res},t}} |\mathcal S|.
\]
The resulting set is identical to that produced by the high-complexity sorting-based greedy rule, which selects packets in nondecreasing order of $\alpha s_j \overline{P}_{b,j,t}$ until the loss budget is exhausted.
\end{Lemma}
\begin{proof}
    (See Appendix~\ref{lemma3proof}).
\end{proof}

\begin{algorithm}[t]
\caption{On-Device Retransmission Control (Round $t$)}
\label{alg:OnDevice}
\begin{algorithmic}[1]
\STATE \textbf{Input:} $\mathcal V_t$, $\beta_{\mathrm{res},t}$, $\{s_j\}$, $\{P_{b,j,i}\}_{i=1}^t$;
\STATE \textbf{Initialize:} $\overline{P}_{b,j,t}\!\leftarrow\!\frac{1}{t}\sum_{i=1}^t P_{b,j,i}$ for all $\mathcal U_j\!\in\!\mathcal V_t$, $\ell\!\leftarrow\!1$, $\mathcal V_{t,1}\!\leftarrow\!\mathcal V_t$, $\beta_{\mathrm{res},t,1}\!\leftarrow\!\beta_{\mathrm{res},t}$;
\WHILE{$\mathcal V_{t,\ell}\neq\emptyset$}
    \STATE $\mathcal T_\ell \leftarrow \{\mathcal U_j\!\in\!\mathcal V_{t,\ell}: \overline{P}_{b,j,t} \le \frac{\beta_{\mathrm{res},t,\ell}}{\alpha s_j |\mathcal V_{t,\ell}|}\}$;
    \IF{$\mathcal T_\ell=\emptyset$} \STATE \textbf{break}; \ENDIF
    \STATE $\beta_{\mathrm{res},t,\ell+1} \leftarrow \beta_{\mathrm{res},t,\ell}-\sum_{\mathcal U_j\in\mathcal T_\ell}\alpha s_j \overline{P}_{b,j,t}$;
    \STATE $\mathcal V_{t,\ell+1} \leftarrow \mathcal V_{t,\ell}\setminus \mathcal T_\ell$, \quad $\ell\leftarrow \ell+1$;
\ENDWHILE
\STATE $\mathcal V'_t\leftarrow \mathcal V_{t,\ell}$, \ $\beta'_{\mathrm{res},t}\leftarrow \beta_{\mathrm{res},t,\ell}$, \ $c_j\leftarrow \alpha s_j \overline{P}_{b,j,t}$;
\STATE Sort $\mathcal V'_t$ by nondecreasing $c_j$; \quad $\mathcal S_t\leftarrow\emptyset$; \ $s\leftarrow 0$;
\FORALL{$\mathcal U_j$ in sorted $\mathcal V'_t$}
    \IF{$s+c_j \le \beta'_{\mathrm{res},t}$} \STATE $\mathcal S_t\leftarrow \mathcal S_t\cup\{\mathcal U_j\}$; \ $s\leftarrow s+c_j$;
    \ELSE \STATE \textbf{break}; \ENDIF
\ENDFOR
\STATE $\mathcal V'_t \leftarrow \mathcal V'_t\setminus \mathcal S_t$; \quad $\beta_{\mathrm{res},t+1}\leftarrow \beta'_{\mathrm{res},t}-\sum_{\mathcal U_j\in\mathcal S_t}c_j$;
\STATE \textbf{Output:} ACK for $\mathcal V_t\setminus \mathcal V'_t$, NACK for $\mathcal V'_t$.
\end{algorithmic}
\end{algorithm}

The proposed retransmission control algorithm has two key advantages.
\begin{enumerate}
    \item \textit{Parametric-sensitivity-based control:}
    The stopping condition \eqref{termination_con} is assigned based on the joint impact of packet-level parametric sensitivity $s_j$ and the instantaneous channel conditions, thereby enabling distinct reliability control across packets. Consequently, packets containing highly sensitive parameters (i.e., large $s_j$) are subject to stricter BER constraints, ensuring high reliability for these parameters that most affect model performance. In contrast, packets with lower cumulative sensitivity are allowed higher BER thresholds to avoid redundant retransmissions. This sensitivity-aware adaptation significantly reduces the total retransmission overhead while preserving downloaded model performance, thus enhancing the efficiency of model downloading.
    \item \textit{Packet elimination with reduced on-device complexity:} 
    In Phase 1, the device performs parallel, per-packet threshold comparisons that pre-eliminate a large fraction of packets without global ordering. 
    In Phase 2, the remaining budget is allocated to the reduced set by sequentially selecting packets with the smallest $\alpha s_j \overline{P}_{b,j,t}$, thereby matching the termination decisions of the sorting-based greedy algorithm. This design avoids global sorting in favor of parallel comparisons and reductions, yielding equivalent termination decisions.
\end{enumerate}

\begin{Remark}[Comparison with Channel-Aware Retransmission]\rm
We compare PASAR with conventional channel-aware retransmission, which uses a fixed BER threshold uniformly across packets\cite{sheng2017performance,chasecombining}:\vspace{-2.5mm}
    \begin{align}\label{traditionalARQ}
        \overline{P}_{b,j,t}\leq \frac{\beta_\text{total}}{\alpha\sum_{j=1}^{J}s_j}.
    \end{align}
The BER requirement in \eqref{traditionalARQ} results in a uniform reliability target for all packets. While simple to implement, this approach incurs unnecessary retransmissions for less critical parameters and potentially inadequate retransmissions for sensitive ones. In contrast, the proposed parametric-sensitivity-aware ARQ in \eqref{termination_con} introduces heterogeneous reliability, adapting the retransmission effort to the varying impact of different parameters on model performance. As a result, this approach enables more efficient use of radio resources, and thus accelerates the AI-model downloading process.
\end{Remark}
\vspace{-3mm}
\subsection{Complexity Analysis}
In Phase 1, packets are terminated by comparing their instantaneous BER observation against a dynamic threshold scaled by their own parametric sensitivity. Under a strongly positively skewed sensitivity distribution as shown in Fig.~\ref{skewness_ditribution}, the mean typically exceeds the median\cite{skewnessmedianmean}. As a result, in each intra-round update, more than half of the active packets satisfy the termination condition, producing a rapid geometric decay in the number of packets that require further retransmission. An upper bound is given by the worst case in which the number of active packets is halved at each iteration, yielding the sequence $J + \frac{J}{2} + \frac{J}{4} + \dots + 1$ and an overall linear complexity of $O(J)$. In Phase 2, the refinement of the residual set $\mathcal V'_t$ involves sorting the packets according to their costs, which incurs a time complexity of $O(|\mathcal V'_t|\log|\mathcal V'_t|)$. Consequently, the overall computational complexity of the PASAR is $O(J)+O(|\mathcal V'_t|\log|\mathcal V'_t|)$. Let $K$ denote the number of intra-round updates in Phase 1. 
In the worst case where the active set is halved at each update, the residual set satisfies $|\mathcal V'_t|\le J/2^K$, and thus the refinement cost is upper bounded by $\mathcal{O}\!\left(\frac{J}{2^K}\log\frac{J}{2^K}\right)$, which is a $2^K$-fold reduction relative to global sorting $\mathcal{O}(J\log J)$.

To have a concrete understanding, consider conventional local downloading models with $N=10^6\!\sim\!10^7$ parameters \cite{sandler2018mobilenetv2,ma2018shufflenet}, quantized to 8 bits and segmented into 1000-bit packets, we have $J\approx10^3\!\sim\!10^4$. Under the derived complexity bound, the PASAR stopping module incurs at most $10^4 \sim 10^5$ floating-point operations (FLOPs). Due to the strongly positively skewed sensitivity distributions observed in practice, the active packet set shrinks much more aggressively. The empirical measurements show that the actual FLOPs per round remain below $10^3$. 
On modern mobile or edge devices with clock frequencies in the GHz range\cite{GHz}, executing $10^5$ (worst case) or $10^3$ (in practice) operations incurs a latency on the order of $10^{-4}$ (worst case) or $10^{-6}$ (in practice) seconds, rendering the computational overhead of PASAR negligible.

\vspace{-3mm}
\section{Extension: PASAR with Parameter Pruning}
The highly skewed distribution of parametric sensitivity (see Section III-B) motivates parameter pruning to reduce communication overhead for more efficient AI downloading. In this section, we discuss the extension of the PASAR protocol to incorporate the parameter pruning. This reveals the underpinning trade-off between sensitivity heterogeneity, parameter pruning, and inference accuracy. 

Pruning works by removing parameters with low sensitivity, which, as revealed in our analysis, have a minimal impact on inference accuracy. 
Thus, moderate pruning benefits the efficiency of AI model downloading systems as it reduces the total number of parameters for transmission and thus retransmission latency while preserving the inference performance of the downloaded model.
On the other hand, overly aggressive pruning may discard parameters with non-negligible sensitivities, degrading inference accuracy and potentially violating the target performance requirement. \emph{These effects give raise to a latency–accuracy tradeoff}. From the perspective of PASAR, pruning reduces the skewness in the parametric sensitivity distribution and hence the performance gains from PASAR. 
Nevertheless, even with pruning, PASAR continues to outperform conventional channel-aware ARQ schemes with uniform stopping rules, provided that the remaining parameters retain a positive skewness in their sensitivities. The claim is substantiated by experimental results in Section VII-C.
\vspace{-2mm}
\section{Experimental Results}\label{ExperimentResult}
\subsection{Experimental Settings}
The default experimental settings are as follows unless specified otherwise.
\begin{itemize}
\item \textit{DNN model and datasets:}
In this study, the LeNet-5 \cite{Lenet5} and ShuffleNetV2 \cite{ma2018shufflenet} models are employed for model downloading tasks. The LeNet-5 model is trained on the well-established MNIST \cite{mnist} dataset, which contains handwritten digit images. The dataset is split into training and test sets, comprising 60,000 and 10,000 samples, respectively. On the other hand, ShuffleNetV2 is trained on the CIFAR-10 dataset \cite{cifar10}, with its training and test sets consisting of 50,000 and 10,000 samples, respectively. For both models, training is conducted over 100 epochs, and the optimal model is selected based on the highest test accuracy achieved.

\item \textit{Communication settings:}
During model downloading, the downlink is modeled as a Rayleigh fading channel. For each packet transmission, the channel coefficient follows i.i.d. $\mathcal{CN}(0,1)$. The model parameters are divided into different packets based on their order. For the purpose of comparison in the experiments, we fix the packet payload to 500 or 1000 information bits. The system operates with a bandwidth of 20 MHz. 
For channel coding, we select low-density parity-check (LDPC) codes\cite{LDPC}, and the modulation order and code rate for the benchmarks are chosen according to the MCS table under a specific SNR condition. Each model parameter is quantized with $n=8$ bits. The maximum number of transmissions is set to $T_\text{max}=25000$. The downloading latency is measured by averaging the total model downloading time across 5,000 independent simulation runs.

\item \textit{Stopping-control and budget settings:}
PASAR is implemented exactly according to Algorithm~\ref{alg:SystemPASAR} and Algorithm~\ref{alg:OnDevice}. 
For each packet $\mathcal U_j$, the receiver computes the empirical average BER at round $t$ as $\overline{P}_{b,j,t}=\frac{1}{t}\sum_{i=1}^t P_{b,j,i}$, where $P_{b,j,i}$ is obtained by counting bit errors over the packet payload. 
Within each round, PASAR applies the termination condition in \eqref{termination_con} with the simultaneous budget/active-set updates in \eqref{update1}--\eqref{update2} until no packet satisfies \eqref{termination_con}. The greedy refinement phase is then executed by sorting the remaining packets according to $\alpha s_j\overline{P}_{b,j,t}$ in ascending order and terminating the smallest ones until the residual budget is exhausted. The resulting residual budget is carried to the next round. For a given model, we set $\beta_{\mathrm{total}}$ as the cumulative loss when the assembled model reaches the target accuracy: $95\%$ for LeNet-5 (MNIST) and $85\%$ for ShuffleNetV2 (CIFAR-10).

\end{itemize}

Three well known retransmission schemes are used to benchmark the performance of the PASAR design.
\begin{itemize}
    \item \textit{HARQ-Type I}\cite{sheng2017performance}:
        This basic HARQ scheme retransmits the same data upon receiving a NACK, without incorporating additional redundancy or combining past transmissions.
        \item \textit{HARQ-Chase Combining (CC)}\cite{chasecombining}:
       Upon a failed transmission, HARQ-CC retransmits the same encoded data block. The receiver coherently combines the newly received signal with previously stored versions, enhancing decoding reliability through signal diversity. 
        \item \textit{HARQ-Incremental Redundancy (IR)}\cite{chasecombining}:
         HARQ-IR extends the retransmission mechanism by sending new redundancy bits upon each retransmission attempt. Rather than repeating the same data, each retransmission provides additional parity information, which is accumulated at the receiver for joint decoding. This progressive refinement enhances error correction capability and improves the spectral efficiency compared with repetition-based schemes.
    \end{itemize}
    For all HARQ baselines, retransmissions are controlled by a uniform stopping threshold. 
    Specifically, packet $\mathcal U_j$ is declared successfully received (ACK) once its empirical average BER satisfies $\overline{P}_{b,j,t} \le \frac{\beta_{\mathrm{total}}}{\alpha\sum_{j=1}^{J}s_j}$, and otherwise a NACK is returned and the packet will be retransmitted.
    \begin{figure}[t!]
    \centering
    \subfigure[1000 information bits per packet.]{
    \includegraphics[width=0.6\columnwidth]{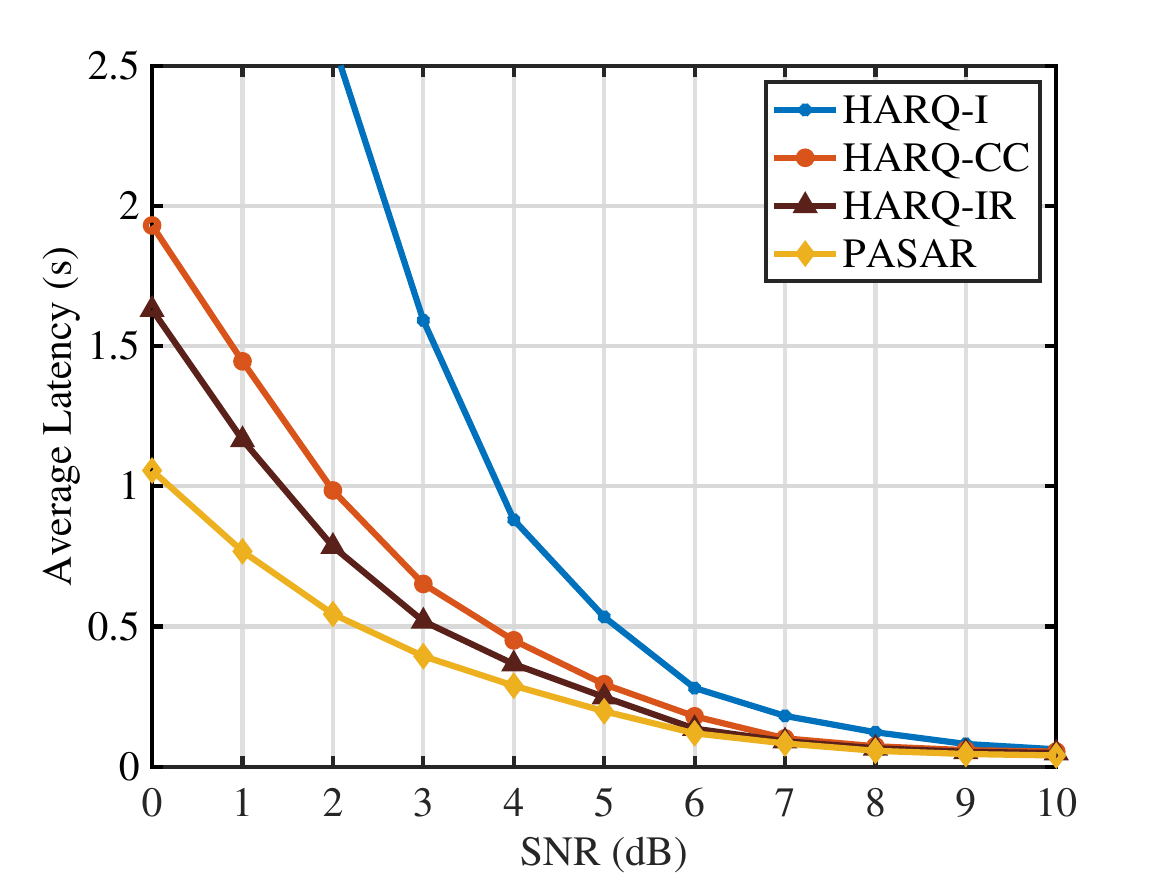}\label{LeNet_200}}\vspace{-1mm}
    \subfigure[500 information bits per packet.]{
    \includegraphics[width=0.6\columnwidth]{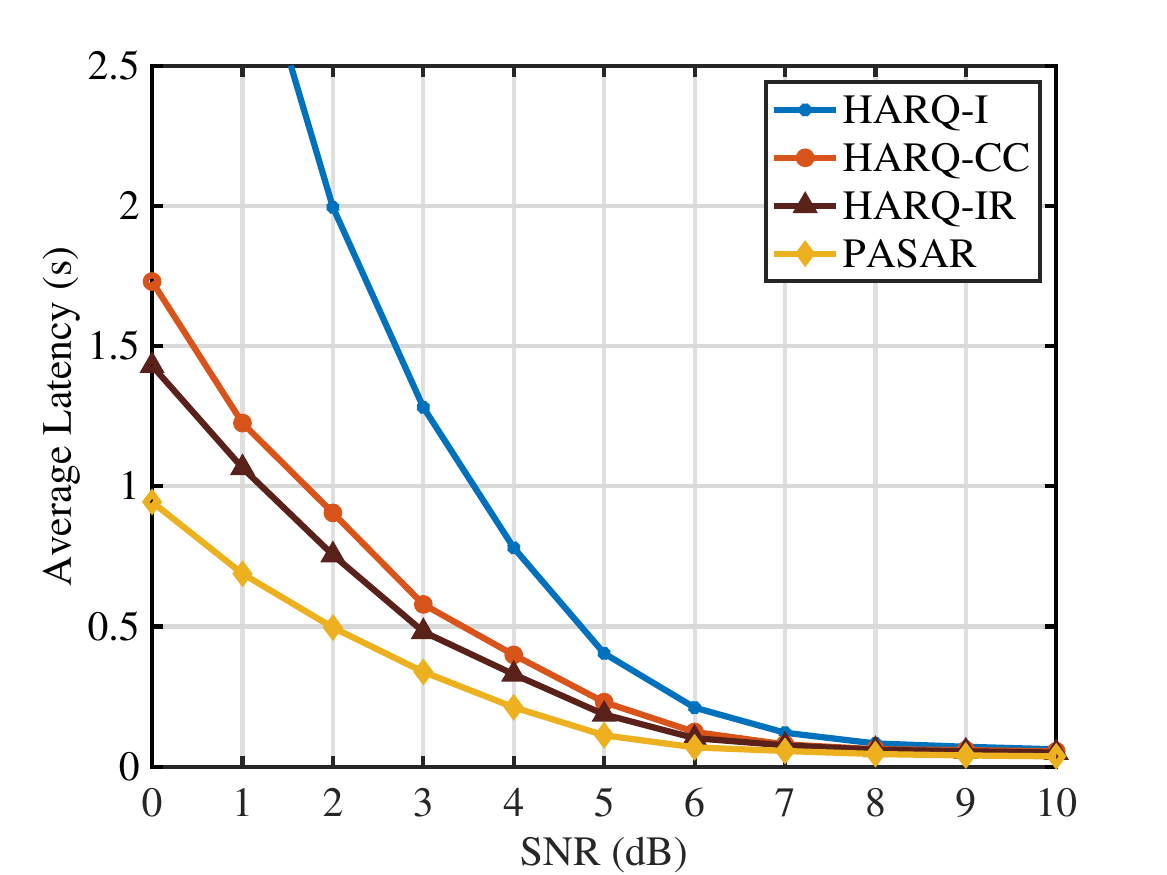}\label{LeNet_100}}\vspace{-1mm}
    \caption{AI downloading latency versus SNR on MNIST using LeNet.}
    \label{ARQ_VS_Sensitivity_LeNet}\vspace{-6mm}
    \end{figure}
    \vspace{-2mm}
    \subsection{Performance of PASAR Protocol}
    We start by presenting the results of the average retransmission latency across different SNR values for the various HARQ schemes, as shown in Fig.\ref{ARQ_VS_Sensitivity_LeNet}. The experiments were conducted on the MNIST dataset using LeNet, with two different configurations: 1000 information bits per packet (Fig.\ref{LeNet_200}) and 500 information bits per packet (Fig.\ref{LeNet_100}).  The proposed PASAR scheme is evaluated against three baselines: HARQ-I, HARQ-CC, and HARQ-IR. At low SNR values, where transmission errors are frequent, PASAR significantly outperforms conventional HARQ schemes. This performance gain arises from its ability to differentiate parameters based on their sensitivity to model performance, enabling retransmissions to focus on critical parameters while avoiding unnecessary retransmissions for less important ones. In contrast, traditional HARQ methods apply a uniform retransmission policy, resulting in redundant transmissions and increased latency. As SNR improves and the likelihood of bit errors decreases, the overall number of retransmissions decreases, narrowing the performance gap between PASAR and the benchmarks. In these high-SNR conditions, the system inherently benefits from more reliable communication, reducing the marginal advantage of sensitivity-aware retransmission.
    
    The comparative results demonstrate that the 500-information-bit packet configuration yields lower average latency and faster convergence compared with the 1000-information-bit configuration. This improvement is attributed to the reduced packet size, which shortens each transmission and decoding cycle, accelerates error detection and correction, and allows retransmission loops to complete more rapidly. As a result, the system benefits from tighter feedback control, more efficient utilization of communication resources, and faster model synchronization across all SNR levels.
    
    \begin{figure}[t!]
    \centering
    \subfigure[1000 information bits per packet.]{
    \includegraphics[width=0.6\columnwidth]{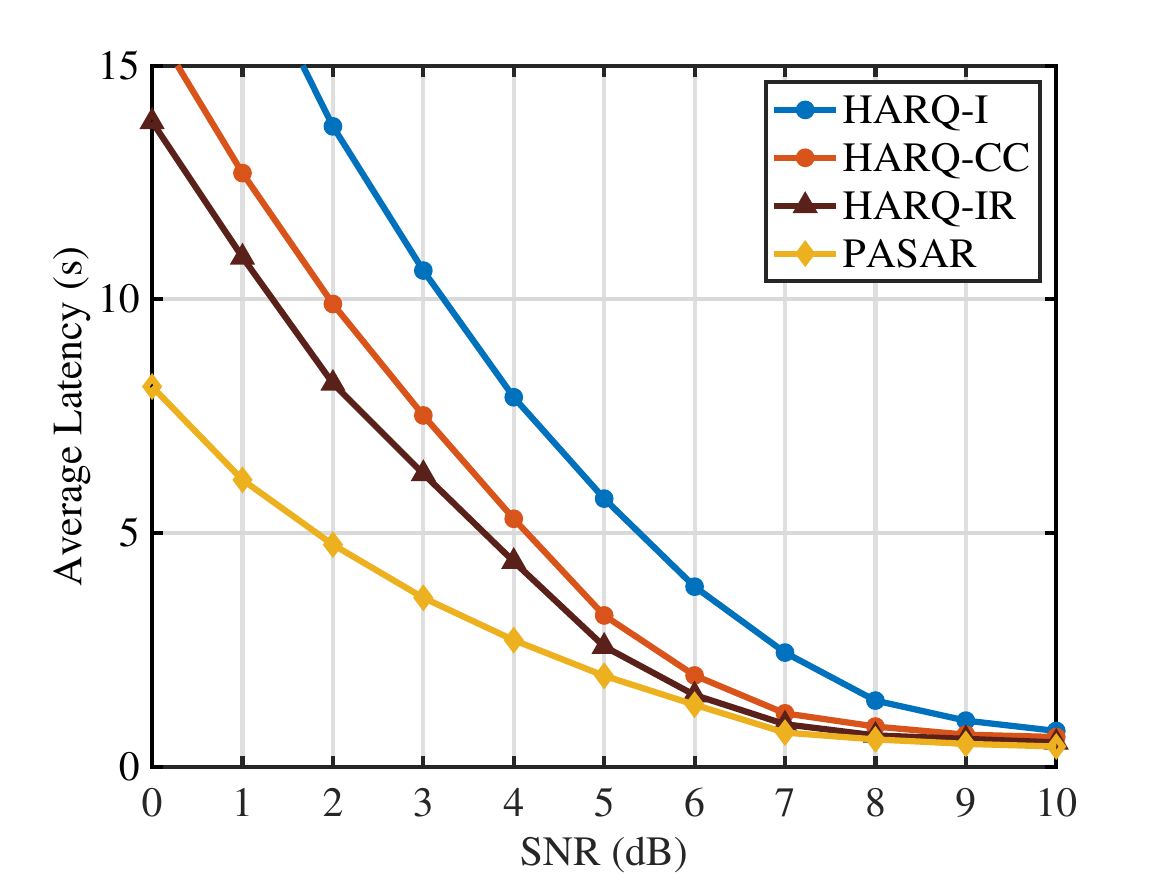}\label{shufflenet_200}}\vspace{-1mm}
    \subfigure[500 information bits per packet.]{
    \includegraphics[width=0.6\columnwidth]{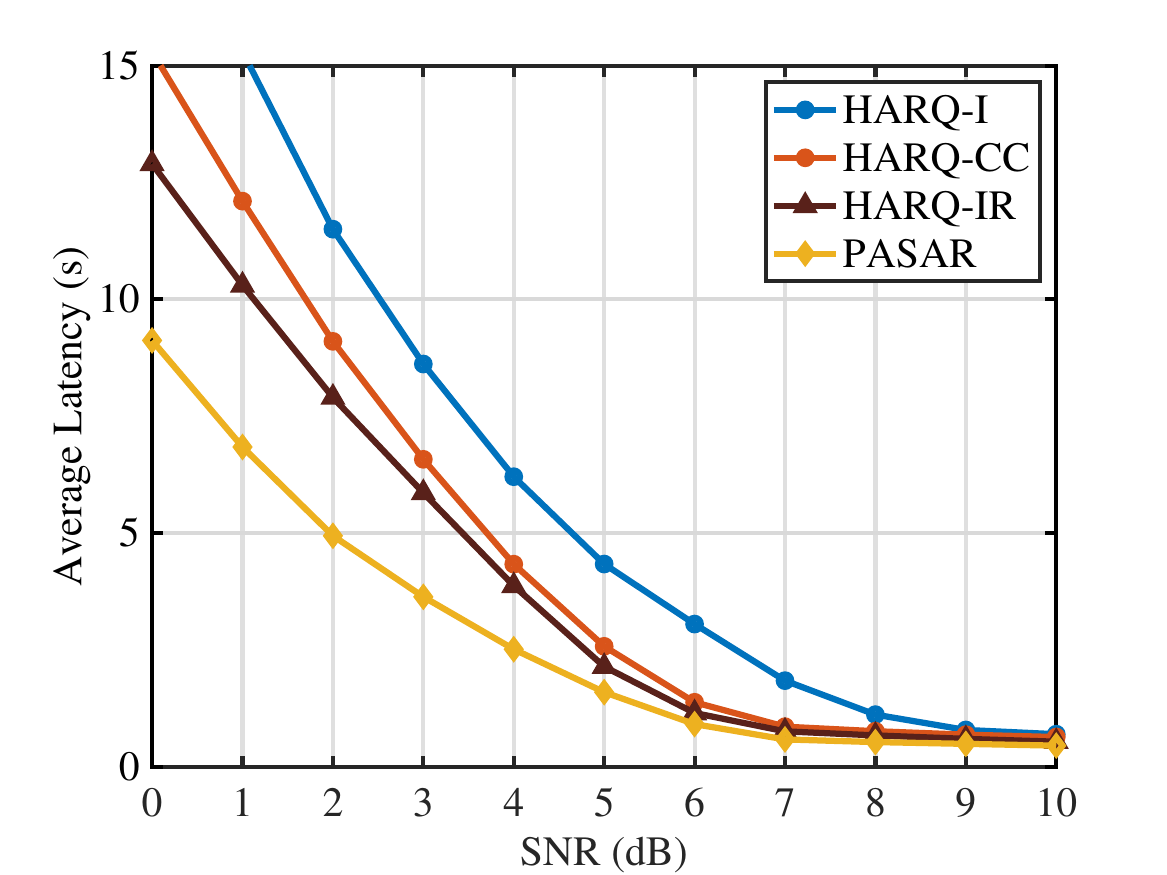}\label{shufflenet_100}}\vspace{-1mm}
    \caption{AI downloading latency versus SNR on Cifar10 using ShuffleNetV2.}
    \label{ARQ_VS_Sensitivity_ShuffleNet}\vspace{-6mm}
    \end{figure} 
    
    We then present the results obtained on the CIFAR-10 dataset using ShuffleNetV2, as shown in Fig.~\ref{ARQ_VS_Sensitivity_ShuffleNet}. Comparing with the MNIST experiments, the average retransmission latency is significantly higher, which can be attributed to the considerably larger parameter set of ShuffleNetV2, resulting in an increase in packet number. In this context, the proposed sensitivity-aware retransmission scheme yields even greater latency reductions relative to HARQ-I, HARQ-CC, and HARQ-IR. This improvement stems from the PASAR framework’s ability to selectively prioritize retransmissions of critical parameters, avoiding unnecessary retransmissions for packets with low sensitivity in high-dimensional models. Notably, as the model size increases and the channel conditions deteriorate, the benefits of PASAR become more pronounced. In such challenging settings, conventional HARQ schemes incur significant overhead due to indiscriminate retransmissions, whereas PASAR maintains efficiency by exploiting parameter importance, resulting in substantial latency reduction.
    
    \begin{figure}[t!]
    \centering
    \includegraphics[width=0.6\columnwidth]{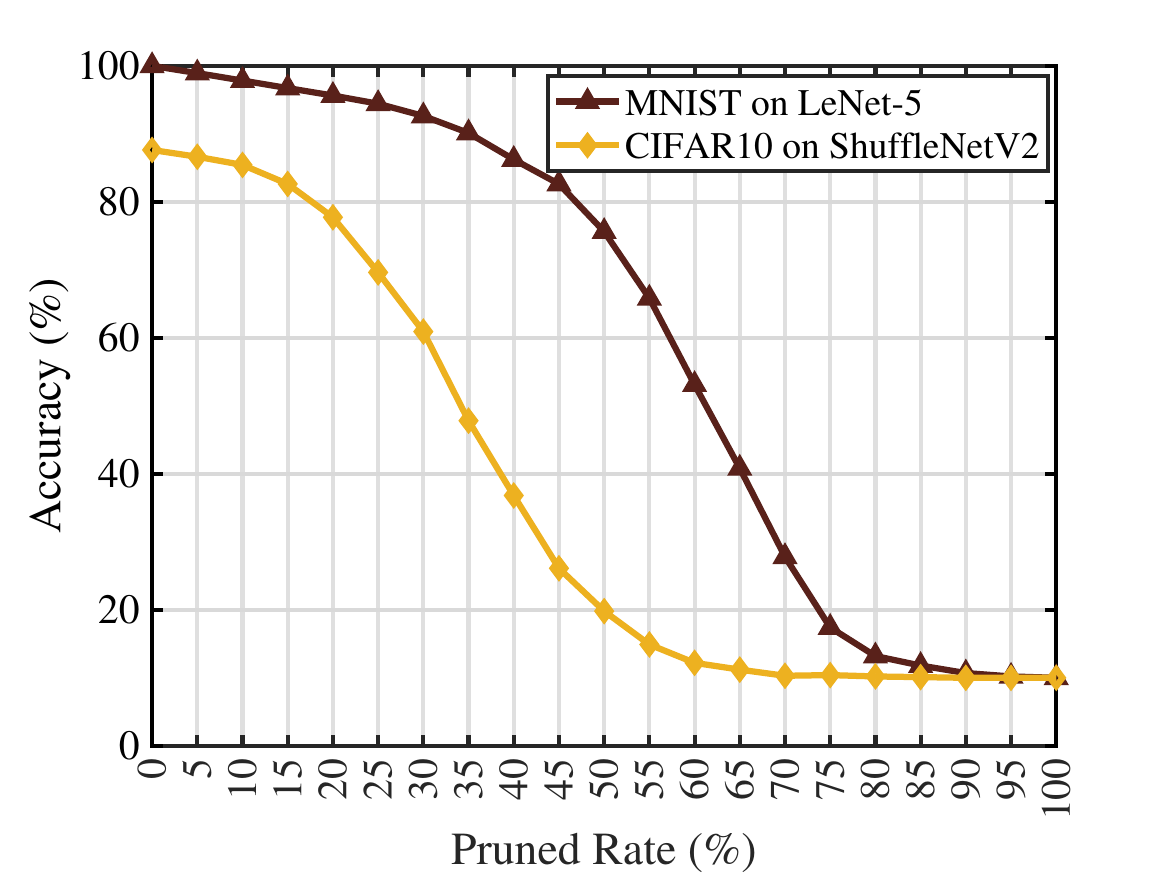}\vspace{-3mm}
    \caption{Model's inference accuracy under different pruning rate.}
    \label{model_accuracy_vs_Pruning}\vspace{-6mm}
    \end{figure}
    \vspace{-3mm}
    \subsection{Effects of Sensitivity Distribution Skewness}
    To investigate how the distribution of parametric sensitivity influences the effectiveness of PASAR, we simulate varying degrees of parameter pruning as a proxy to modulate the sensitivity landscape. As the pruning rate increases, the low-sensitivity parameters are removed, leading to a reduction in the skewness of the sensitivity distribution and a more homogeneous set of remaining parameters. This allows us to systematically examine how the performance gains of PASAR, which arise from its selective retransmission mechanism, are affected by the degree of sensitivity skewness across model parameters. We begin by identifying the maximum allowable pruning rates for each dataset that preserve acceptable model accuracy, which serve as constraints for subsequent evaluations of our PASAR scheme against traditional HARQ baselines. It is important to note that HARQ-I is excluded from this comparison because its significantly higher baseline latency results in disproportionately large performance gains, making it difficult to visualize and interpret meaningful differences in the figure. The packet payload is fixed to 1000 information bits. As shown in Fig.~\ref{model_accuracy_vs_Pruning}, LeNet-5 on MNIST maintains accuracy above 95\% with pruning rates up to 20\%, while ShuffleNetV2 on CIFAR-10 retains accuracy above 85\% up to 10\% pruning. Accordingly, we adopt 20\% and 10\% as the maximum pruning thresholds for the MNIST and CIFAR-10 experiments, respectively. These thresholds ensure that pruning-induced degradation remains within acceptable limits, allowing a fair and effective evaluation of the latency benefits offered by the proposed retransmission framework.
    
    \begin{figure}[t!]
    \centering
    \subfigure[Comparison to HARQ-CC.]{
    \includegraphics[width=0.6\columnwidth]{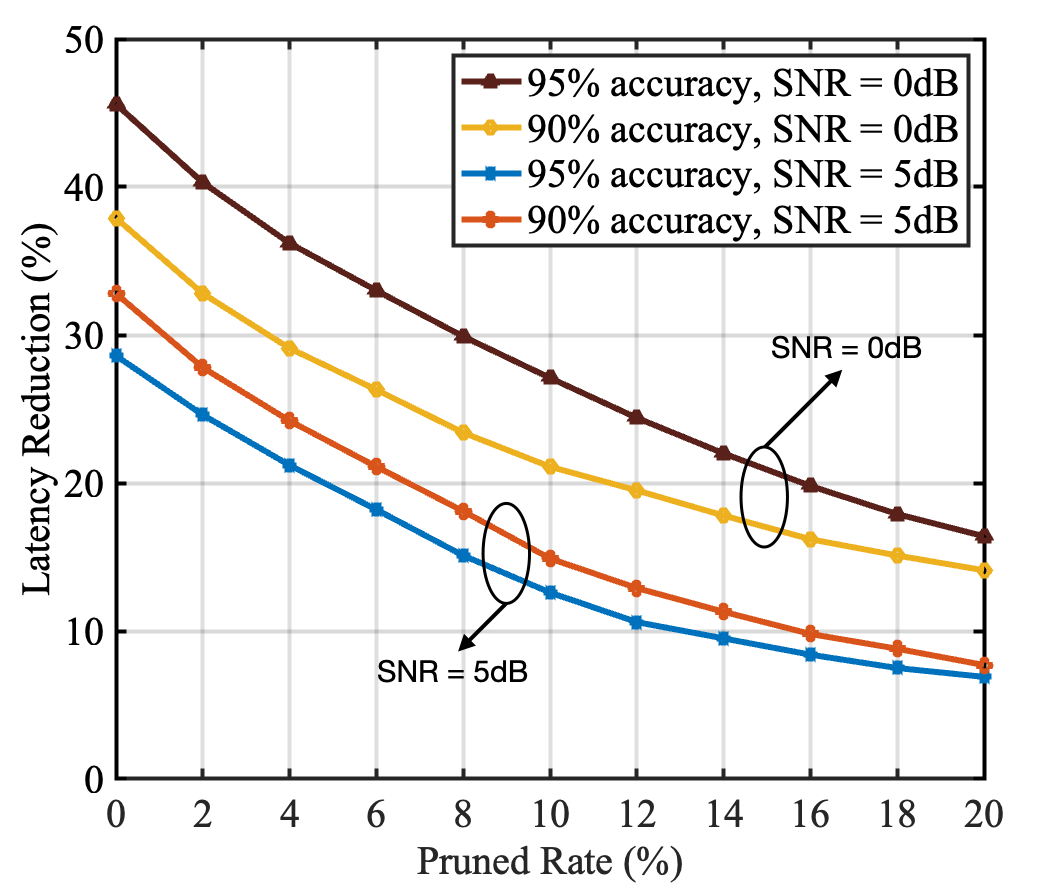}\label{Pruned_Lenet_CC}}\vspace{-1mm}
    \subfigure[Comparison to HARQ-IR.]{
    \includegraphics[width=0.6\columnwidth]{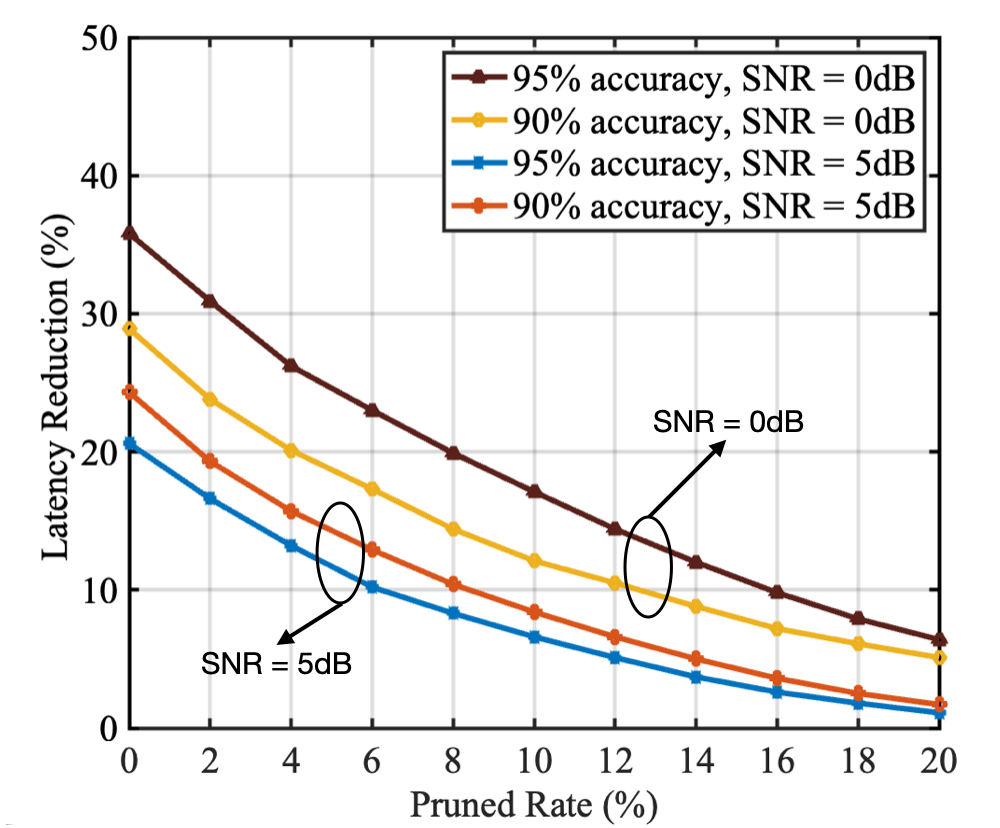}\label{Pruned_Lenet_IR}}\vspace{-1mm}
    \caption{Reduction of retransmission latency versus pruning rate on MNIST with 1000 information bits per packet.}
    \label{Pruned_retransmission_Lenet}\vspace{-6mm}
    \end{figure}
    Figs.~\ref{Pruned_Lenet_CC} and \ref{Pruned_Lenet_IR} illustrate the latency reduction achieved by the proposed PASAR scheme compared with HARQ-CC and HARQ-IR at different pruning rates. At low pruning levels, the sensitivity distribution remains highly skewed. In this case, a small subset of parameters dominates the model's robustness, allowing PASAR to selectively allocate retransmission efforts toward the most critical parameters. In this regime, PASAR achieves substantial latency savings, reducing retransmission cost by approximately 45\% over HARQ-CC and 35\% over HARQ-IR (with 95\% accuracy, 0 dB SNR).
    However, as the pruning rate increases, low-sensitivity parameters are progressively removed, flattening the sensitivity distribution. The remaining parameters tend to exhibit uniformly high importance, thereby reducing the heterogeneity that PASAR leverages for prioritization. Consequently, the retransmission strategy transitions from a skew-aware, importance-driven scheme to a near-uniform policy, leading to diminishing performance gains. At a pruning rate of 20\%, the improvements drop to roughly 17\% (vs. HARQ-CC) and 8\% (vs. HARQ-IR) under 5 dB SNR. 
    These observations highlight that the latency advantage of PASAR is intrinsically tied to the degree of skewness in the sensitivity distribution. Greater heterogeneity in parameter importance provides more opportunities for selective retransmission, thus allowing larger communication latency reduction.
    
    \begin{figure}[t!]
    \centering
    \subfigure[Comparison to HARQ-CC.]{
    \includegraphics[width=0.6\columnwidth]{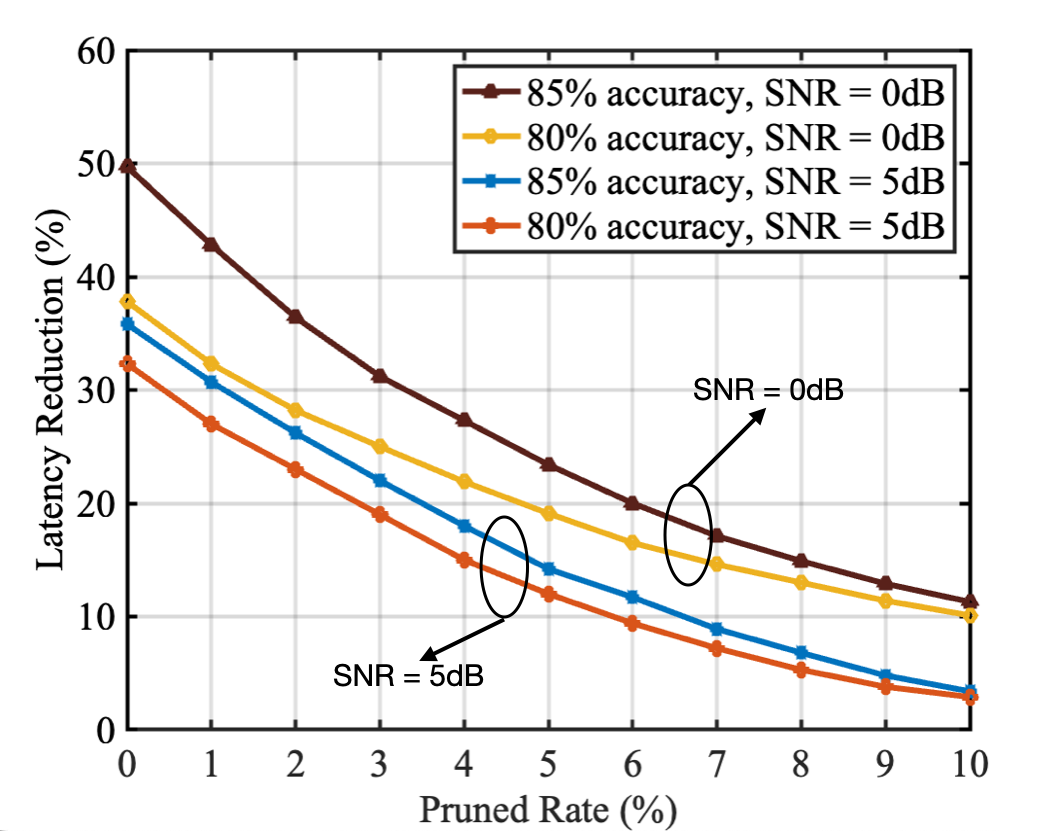}\label{Pruned_Shufflenet_CC}}\vspace{-1mm}
    \subfigure[Comparison to HARQ-IR.]{
    \includegraphics[width=0.6\columnwidth]{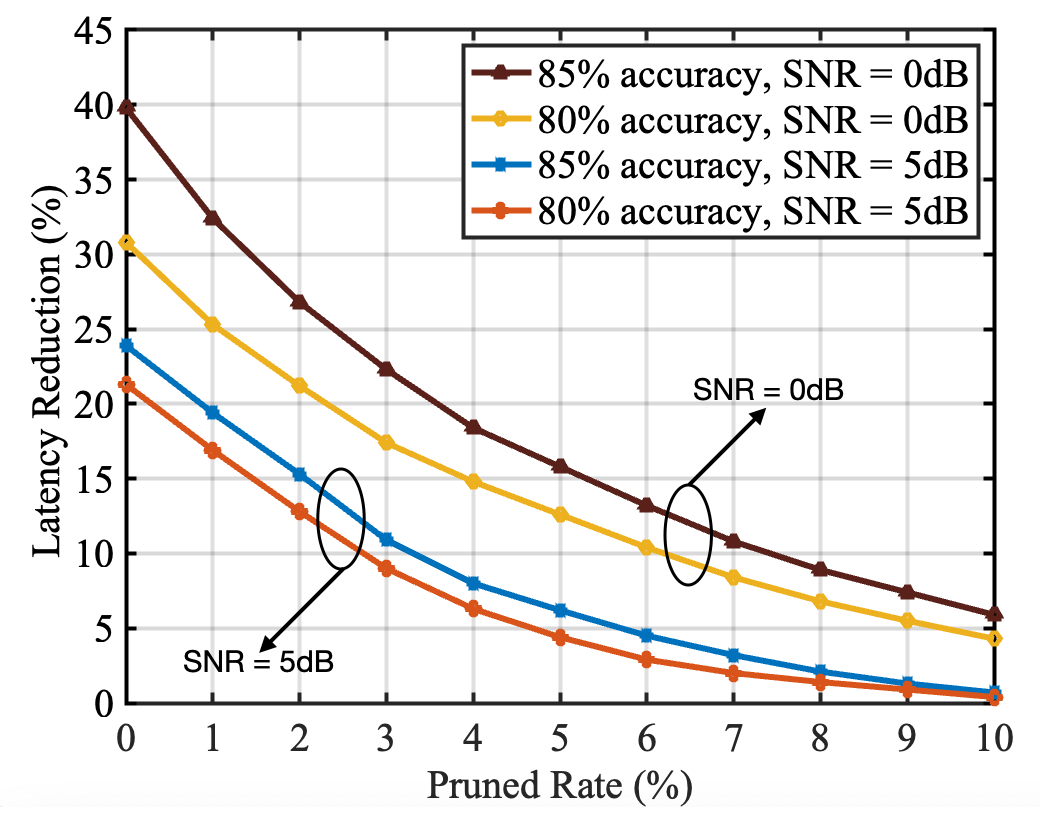}\label{Pruned_Shufflenet_IR}}\vspace{-1mm}
    \caption{Reduction of latency versus pruning rate on CIFAR10 with 1000 information bits per packet.}
    \label{Pruned_retransmission_Shufflenet}\vspace{-6mm}
    \end{figure}
    Fig.~\ref{Pruned_Shufflenet_CC} and Fig.~\ref{Pruned_Shufflenet_IR} present the latency performance of the proposed sensitivity-aware scheme compared with HARQ-CC and HARQ-IR in the CIFAR-10 dataset using ShuffleNetV2 at different pruning rates. Consistent with the MNIST results, the sensitivity scheme achieves progressively greater latency reductions as the pruning rate increases, outperforming traditional HARQ schemes across the board. Notably, the performance gains on CIFAR-10 are more significant, particularly at 0 dB, than those observed on MNIST. This improvement comes from the substantially larger dimension of ShuffleNetV2, which allows a greater number of low-sensitivity parameters to be pruned. As a result, the sensitivity scheme can eliminate a larger number of unnecessary retransmissions, resulting in more significant latency reductions. In contrast, the smaller parameter count of LeNet-5 in MNIST limits the extent of pruning, thus constraining the achievable latency improvement.
    
    As the target accuracy increases, the number of required retransmissions naturally grows, resulting in higher communication latency. Notably, the performance gain of PASAR becomes more pronounced when the accuracy threshold increases from 80\% to 85\% at 0 dB SNR, compared with the same increase at 5 dB. This observation underscores the particular effectiveness of PASAR under challenging channel conditions and stringent performance requirements.
    The key advantage of PASAR lies in its ability to exploit skewness in the parametric sensitivity distribution, where a small subset of parameters dominates model performance. While this skewness is inherent to the model, particularly evident at low pruning rates, its impact on communication efficiency is amplified under adverse conditions such as low SNR, large model dimensionality, or strict reliability requirements. In such settings, traditional uniform retransmission becomes increasingly inefficient. PASAR mitigates this by prioritizing retransmissions for highly sensitive parameters, yielding substantial latency reductions.
    \vspace{-3mm}
    \section{Conclusion}
    In this paper, we present a goal-oriented design of retransmission, namely the PASAR protocol, to support reliable, low latency AI downloading applications. 
    It features a retransmission-control algorithm that exploits the highly skewed distribution of parametric sensitivity to reduce communication latency via prioritizing transmission of critical parameters (with high sensitivity).
    The idea is materialized through designing adaptive stopping thresholds that are aware of parametric sensitivity.
    Through such awareness, the proposed PASAR protocol is shown to significantly outperform traditional retransmission schemes that exploits only channel state information.
    
    To the best of our knowledge, this is the first work to investigate the retransmission framework for AI-model downloading by taking  parametric sensitivity into account. This opens several promising avenues for future research. One direction involves jointly optimizing retransmission-control thresholds and pruning rates (see Section VI) to balance accuracy and latency, which requires deeper investigation into the combined effects of pruning and quantization on these two metrics. Another important extension is to incorporate energy efficiency by designing adaptive power control schemes that consider both channel fading and parametric sensitivity. Finally, extending this framework to distributed edge inference scenarios, potentially integrating over-the-air computation, could further enhance the system scalability and efficiency via sensitivity-awareness.
    \vspace{-6mm}
    \appendix
    \subsection{Proof of Lemma \ref{Lemma0}}\label{Lemma0proof}
    The received distorted value $\hat{w}_d$ is given by:
    \begin{align}
    \hat{w}_d = &\sum_{i=0}^{n-2}[a_i(-1)^{b_i} + b_i]2^i - [a_{n-1}(-1)^{b_{n-1}} + b_{n-1}]2^{n-1},\nonumber
    \end{align}
    where $a_i$ are independent Bernoulli random variables, $a_i \sim \text{Bernoulli}(P_b)$, and $b_i \in \{0, 1\}$ are the original bits.
    The error for the $d$-th parameter is defined as:
    % \begin{equation}
    $\Delta {w}_d = \hat{w}_d - w_d,$
    % \end{equation}
    where $w_d$ is the true value of the parameter:
    \begin{equation}
    w_d = \sum_{i=0}^{n-2} b_i 2^i - b_{n-1} 2^{n-1}.
    \end{equation}
    The error in the received value is:
    \begin{equation}
    \Delta {w}_d = \sum_{i=0}^{n-2} a_i (-1)^{b_i} 2^i - a_{n-1} (-1)^{b_{n-1}} 2^{n-1}.
    \end{equation}
    Now we calculate the variance of $\Delta {w}_d$. The variance for each bit error $a_i$ is:
    % \begin{equation}
    $\mathrm{Var}(a_i (-1)^{b_i} 2^i) = P_b (1 - P_b) (2^i)^2.$
    % \end{equation}
    Since the $a_i$ are independent, the total variance is:
    \begin{align}
    \mathrm{Var}(\Delta {w}_d) &= \sum_{i=0}^{n-2} P_b (1 - P_b) (2^i)^2 + P_b (1 - P_b) (2^{n-1})^2\nonumber\\
    &= P_b (1 - P_b) \left( \sum_{i=0}^{n-2} 4^i + 4^{n-1} \right).
    \end{align}
    The sum of the powers of 4 is a geometric series:
    % \begin{equation}
    $\sum_{i=0}^{n-2} 4^i = \frac{4^{n-1} - 1}{3}.$
    % \end{equation}
    Substituting this into the variance expression, we have
    \begin{align}
    \mathrm{Var}(\Delta {w}_d) &= P_b (1 - P_b) \left( \frac{4^{n-1} - 1}{3} + 4^{n-1} \right)\\
    &= P_b (1 - P_b) \frac{4^n - 1}{3}.
    \end{align}
    For small values of $P_b$, we approximate $P_b (1 - P_b) \approx P_b$. Therefore, we get the final result:
    \begin{equation}
    \mathrm{Var}[\Delta {w}_d] \approx \frac{4^n - 1}{3} P_b.
    \end{equation}
    Here, we follow the fact that the zero mean of the expected variation of the parameter that $\mathbb{E}[\Delta w_d]=0, \quad d\in\{1,\cdots,D\}$.
    Given that each parameter is encoded using 
    $n$-bit representation and subject to a BER $P_b$, the loss of the downloaded model can be derived as: 
    \begin{align} 
    \mathbb{E}\left[\Delta f(\mathbf{w})\right] &\approx \nonumber \frac{1}{2}\sum_{d=1}^{D}\frac{\partial^2 f}{\partial w_d^2}\cdot \mathbb{E}\left[\left(\Delta {w}_d\right)^2\right] \nonumber\\&= \frac{1}{2}\sum_{d=1}^{D}\frac{\partial^2 f}{\partial w_d^2}\cdot \Big\{\mathbb{E}\left[\left(\Delta {w}_d\right)^2\right]-\mathbb{E} [\Delta(w_d)]^2\Big\} \nonumber\\&= \sum_{d=1}^{D}\frac{\partial^2 f}{\partial w_d^2}\cdot \frac{4^n-1}{6} P_{b}.
    \end{align}
    Let $\{\mathcal U_j\}_{j=1}^J$ be a partition of $\{1,\ldots,D\}$ into packets and assume all $d\in\mathcal U_j$ share the same packet-level BER $P_{b,j}$. So,\vspace{-1.5mm}
    \begin{align}
    \mathbb{E}[\Delta f(\mathbf w)]
    &\approx \alpha \sum_{j=1}^J \sum_{d\in\mathcal U_j}\frac{\partial^2 f}{\partial w_d^2}\,P_{b,j}
    = \alpha \sum_{j=1}^J s_j P_{b,j},
    \end{align}\vspace{-3mm}
    where $\alpha=\frac{4^n-1}{6}$ and $s_j\triangleq\sum_{d\in\mathcal U_j}\frac{\partial^2 f}{\partial w_d^2}$.
    
    \vspace{-2mm}
    \subsection{Proof of Lemma \ref{Lemma3}}\label{lemma3proof}
    Fix round $t$ and define costs $c_j \triangleq \alpha s_j \overline P_{b,j,t}$ for all $\mathcal U_j\in\mathcal V_t$, with residual budget $B\triangleq \beta_{\mathrm{res},t}$.
    Let $c_{(1)}\le \cdots \le c_{(n)}$ be the sorted costs, $n=|\mathcal V_t|$, and define
    $k^\star \triangleq \max\Big\{k:\ \sum_{i=1}^k c_{(i)}\le B\Big\}.$ For any subset $\mathcal S\subseteq \mathcal V_t$ with $|\mathcal S|=k$, we have
    $\sum_{\mathcal U_j\in\mathcal S} c_j \ge \sum_{i=1}^k c_{(i)}$, since the RHS is the sum of the $k$ smallest costs. Hence, no feasible set can have size $>k^\star$, while the $k^\star$ smallest-cost packets are feasible. Therefore, selecting packets in nondecreasing order of $c_j$ until the budget is exhausted maximizes $|\mathcal S|$.
    
    At epoch $\ell$, let $B_\ell=\beta_{\mathrm{res},t,\ell}$ and $n_\ell=|\mathcal V_{t,\ell}|$, and define $\tau_\ell \triangleq B_\ell/n_\ell$.
    Phase 1 removes all packets with $c_j\le \tau_\ell$ and update $(B_\ell,n_\ell)$.
    Since each removed packet has a cost of at most $\tau_\ell$, their total cost is at most $B_\ell$; any sorting-based greedy rule would select all packets with $c_j\le \tau_\ell$ before any packet with a cost of $>\tau_\ell$ (up to tie-breaking). Hence, removing them yields the same reduced instance as continuing the sorting-greedy approach after taking them.
    Iterating until no $c_j\le \tau_\ell$ produces $(\mathcal V'_t,B')$, and greedy refinement in Phase 2 is finally applied to $\mathcal V'_t$ by selecting the smallest remaining $c_j$ until $B'$ is exhausted, i.e., exactly the continuation of sorting-greedy on the reduced instance.
    Thus, the final termination set is identical to sorting-greedy (with a fixed tie-breaking rule), and it is optimal for maximizing the number of terminations.
    % This completes the proof.
% \vspace{-4mm}
    \bibliographystyle{IEEEtran}
    \bibliography{Ref}
    \end{document}